\documentclass[letterpaper,10pt,journal]{IEEEtran}

\IEEEoverridecommandlockouts

\usepackage{times}
\usepackage{mathptmx}
\usepackage{amsmath}\interdisplaylinepenalty=2500
\usepackage{amssymb}
\usepackage{amsthm}
\usepackage{cite}
\usepackage{url}
\usepackage{color}
\usepackage{graphicx}
\usepackage[amssymb]{SIunits}
\usepackage{bm}
\usepackage{stmaryrd}
\usepackage{steinmetz}
\graphicspath{{figures/png/}}

\def\ket#1{\ensuremath{|#1\rangle}}
\def\bra#1{\ensuremath{\langle #1|}}
\def\ip#1#2{\ensuremath{\langle #1|#2 \rangle}}
\def\qf#1#2#3{\ensuremath{\langle #1|#2|#3 \rangle}}
\def\vec#1{\ensuremath{\mathbf{#1}}}
\def\norm#1{\ensuremath{\|#1\|}}
\def\new#1{#1}

\newcommand{\up}{\ensuremath{\uparrow}}
\newcommand{\dn}{\ensuremath{\downarrow}}
\newcommand{\Sgn}{\ensuremath{\mathrm{sgn}}}
\newcommand{\sinc}{\ensuremath{\mathrm{sinc}}}
\newcommand{\Ne}{\ensuremath{N_e}} 
\newcommand{\Hp}{\ensuremath{\widetilde{H}}}
\newcommand{\Up}{\ensuremath{\widetilde{U}}}
\newcommand{\Pip}{\ensuremath{\widetilde{\Pi}}}
\newcommand{\lambdap}{\ensuremath{\widetilde{\lambda}}}
\newcommand{\omegap}{\ensuremath{\widetilde{\omega}}}
\newcommand{\Np}{\ensuremath{\widetilde{\Ne}}}
\newcommand{\K}{\ensuremath{\mathcal{K}}}
\newcommand{\In}{\ensuremath{\mathrm{in}}}
\newcommand{\Out}{\ensuremath{\mathrm{out}}}

\newcommand{\diag}{\mbox{diag}}

\renewcommand{\v}{\mathbf{v}}
\newtheorem{example}{Example}

\newtheorem{corollary}{Corollary}
\newtheorem{theorem}{Theorem}

\title{\LARGE \bf Design of Feedback Control Laws for Information Transfer in Spintronics Networks}

\author{Sophie G Schirmer$^*$,~\IEEEmembership{Member,~IEEE,} Edmond Jonckheere$^\dag$,~\IEEEmembership{Fellow,~IEEE,} Frank C Langbein$^*$,~\IEEEmembership{Member,~IEEE}%
\thanks{$^*$Supported by the Welsh Government and Higher Education Funding Council for Wales through the S\^{e}r Cymru National Research Network in Advanced Engineering and Materials (NRN082).}%
\thanks{$^\dag$Supported by ARO MURI.}%
\thanks{SGS is with the College of Science (Physics), Swansea University, Swansea, SA2 8PP, UK, \texttt{sgs29@swan.ac.uk}.}%
\thanks{EJ is with the Dept.\ of Electrical Engineering, Univ.\ of Southern California, Los Angeles, CA 90089, USA, \texttt{jonckhee@usc.edu}.}%
\thanks{FCL is with the School of Computer Science \& Informatics, Cardiff University, Cardiff, CF24 3AA, UK, \texttt{LangbeinFC@cardiff.ac.uk}.}%
}

\begin{document}

\maketitle

\begin{abstract}
Information encoded in networks of stationary, interacting spin-1/2 particles is central for many applications ranging from quantum spintronics to quantum information processing. Without control, however, information transfer through such networks is generally inefficient. \new{Currently available control methods to maximize the transfer fidelities and speeds mainly rely on dynamic control using time-varying fields and often assume instantaneous readout.  We present an alternative approach to achieving} efficient, high-fidelity transfer of excitations by shaping the energy landscape via the design of time-invariant feedback control laws without recourse to dynamic control.  \new{Both instantaneous readout and the more realistic case of finite readout windows are considered.  The technique can also be used to freeze information by designing energy landscapes that achieve Anderson localization.} Perfect state or super-optimal transfer and localization are enabled by conditions on the eigenstructure of the system and signature properties for the eigenvectors.  Given the eigenstructure enabled by super-optimality, it is shown that feedback controllers that achieve perfect state transfer are, surprisingly, also the most robust with regard to uncertainties in the system and control parameters.
\end{abstract}

\section{Introduction: Spintronics Devices}

\IEEEPARstart{E}{ncoding} information in spin degrees of freedom has the potential to revolutionize information technology through the development of novel devices utilizing electron spin. Information encoded in spin degrees of freedom can be transferred via spin-polarized currents. Information stored in spin states can also propagate through a network of coupled spins without charge transport, mediated directly by quantum-mechanical interactions. This is of particular interest as devices that do not rely on charge transport are not limited by heat dissipation due to resistance---potentially enabling higher component densities and greater energy efficiency~\cite{spintronics2,spintronics1}.

The realization of novel spintronic devices presents many technological challenges in device design and fabrication. Utilizing information encoded in spin degrees of freedom especially requires efficient, controlled on-chip transfer of excitations in spin networks. In quantum mechanical language, this transfer or transport of an excitation from one site to another requires steering the system from one quantum state to another, a problem akin to the well known unit step response of linear Single Degree of Freedom (SDoF) tracking controllers---with the significant difference of the presence of a global phase factor in the tracking error. As propagation of spin-based information is fundamentally governed by quantum-mechanics and the Schr\"odinger equation, however, excitations in a spin network propagate, disperse and refocus in a wave-like manner. Controlling information transport in such networks is thus a highly non-classical control problem. Previous work has shown that natural transmission of information does occur, but without active control the propagation of spin-based information in such networks can be slow and inefficient~\cite{rings_QINP}.

In this paper we consider how we can optimize transport in terms of transfer efficiency, speed and robustness using control. This requires an approach quite different from modern robust control, where time-domain specifications are substituted for conventional singular value Bode plots. The need for state-selective transfer makes the architecture depart from the SDoF configuration and precludes control designs that ensure \emph{asymptotic stability} of the target state. Instead, we rely on the concept of \emph{Anderson localization}~\cite{Anderson-58,50_years}, which is utilized to hold the system at or around the desired target state for future use.

We explore how information transfer or localization in spin networks can be controlled simply by shaping the energy landscape of the system. We show how the latter problem can be viewed in terms of feedback control laws, and that feedback control designs that achieve the best performance w.r.t. transfer fidelity also achieve the best robustness. This is unlike the traditional limitations observed for SDoF classical control and demonstrates the advantages of two degrees-of-freedom controllers~\cite{2_deg_freedom_controller,Robust_2DOF_MIMO} and is the setup adopted here. The deeper message of this paper is that quantum transport presents many challenges and opportunities for control and a rich source of new problems and paradigms relating to the foundation of classical control theory.

In Section~\ref{sec:theory} relevant theory of quantum spin networks and control paradigms are reviewed. The control objectives, conditions for perfect state transfer and speed limits for excitation transfer are discussed in Section~\ref{sec:design}, followed by eigenstructure analysis of the dynamic generators and signature properties for the eigenvectors to establish general conditions for optimality in Section~\ref{sec:optimality}. In Section~\ref{sec:sensitivity} the sensitivity of the design to uncertainty in the dynamical generators of the system is analyzed, and the result of vanishing sensitivity for superoptimal controllers is proven. Numerical optimization and sensitivity results are presented in Section~\ref{sec:optimization}.  We conclude with a discussion of classical vs quantum robust control in Section~\ref{sec:robust} and general conclusions and directions for future work in Section~\ref{sec:conclusion}.

\section{\label{sec:theory}Theory and Definitions}

\subsection{Networks of Coupled Spins}

Let $X$, $Y$ and $Z$ be the Pauli spin operators
\begin{equation}
  X = \begin{pmatrix} 0 &  1  \\  1 &  0 \end{pmatrix}, \quad
  Y = \begin{pmatrix} 0 & -i  \\  i &  0 \end{pmatrix}, \quad
  Z = \begin{pmatrix} 1 &  0  \\  0 & -1 \end{pmatrix},
\end{equation}
and let $X_k$ ($Y_k$, $Z_k$) be a tensor product of $N$ operators, all of which are the identity $I$, except for a single $X$ ($Y$, $Z$) operator in the $k$th position. With this notation, the Hamiltonian of a system of $N$ spin-$\tfrac{1}{2}$ particles with onsite potentials $D_k$ and two-body interactions between pairs of spins $(k,\ell)$ is
\begin{equation}\label{eq:H}
  H_{\rm full} =  \sum_{k=1}^N D_k Z_k + \sum_{\ell \neq k} J_{k\ell} (X_k X_\ell + Y_k Y_\ell + \kappa Z_k Z_\ell),
\end{equation}
where $J_{k\ell}=J_{\ell{}k}$ for all $k,\ell$ due to the symmetry of the interaction. The constants $D_k$ and $J_{k\ell}$ are measured in units of frequency. $\kappa$ is a parameter that depends on the coupling type: isotropic Heisenberg coupling ($\kappa=1$) or XX coupling ($\kappa=0$). The coupling constants $J_{k\ell}$ are determined by the topology of the network. For a chain with nearest-neighbor coupling we have $J_{k\ell}=0$ unless $k=\ell\pm1$ and similarly for a ring, except that $J_{N,1}=J_{1,N} \neq 0$.  A chain can be thought of as a type of quantum wire and a ring as a basic routing element to distribute information encoded in the network, e.g., via chains attached to nodes of the ring.  A network is \emph{uniform} or \emph{homogeneous} if all non-zero couplings have a fixed strength $J$.  \new{Spin networks of this type are widely applicable to modeling nuclear spin systems, electron spins in quantum dots and pseudo-spin systems consisting of trapped ions or atoms and even superconducting qubits. 
Systems coming very close to reproducing the ideal  dynamics of a one-dimensional Heisenberg chain have been realized~\cite{spintronics2,Heisenberg_Copper,  2016_Heisenberg_simulator, trapped_ions_Heisenberg_simulator_journal, 1D_Heisenberg_spin_chain_simulation}.}

Using the Dirac notation, a (pure) state $\ket{\Psi}$ of a system of $N$ spin-$\tfrac{1}{2}$ particles is a linear combination of the product states of the single spin eigenstates, which are eigenstates of the $Z$ operator denoted by $\ket{\up}$, $\ket{\dn}$:
\begin{equation}
   Z \ket{\up} = +\ket{\up}, \quad Z \ket{\dn} = -\ket{\dn}.
\end{equation}
The operator $Z_k$ applied to a product state $\ket{\dn\cdots\up\cdots\dn}$ thus returns $+1$ if the $k$th spin is $\ket{\up}$, and $-1$ if it is $\ket{\dn}$. Hence, $S=\tfrac{1}{2}\sum_{k=1}^N (I+Z_k)$ effectively counts the number of spins that are in the excited state $\ket{\up}$. The Hamiltonian~\eqref{eq:H} commutes with the total excitation operator, $[H_{\rm full}, S] = 0$. As commuting operators are simultaneously diagonalizable, it can easily be shown that the Hilbert space of the system decomposes into excitation subspaces~\cite{TAC2012} that are invariant under the dynamics. If we assume that only a single excitation (or bit of information) propagates through the network at any given time, then the Hamiltonian can be reduced to the single excitation subspace Hamiltonian
\begin{equation}\label{eq:HS}
  H_{\rm single} = \sum_{k=1}^N (D_k+\kappa J_{k}) \ket{k}\bra{k} + \sum_{\ell > k} J_{k\ell} (\ket{k}\bra{\ell} +\ket{\ell}\bra{k}),
\end{equation}
where the $J_k$ form the diagonal for the single excitation subspace of $\sum_{k\neq \ell} J_{k\ell}Z_kZ_\ell$, which can be absorbed into the $D_k$. $\ket{k}$ can be thought of as a column vector with zero entries except for a $1$ in the $k$th position, $\bra{\ell}$ can be thought of as a row vector with zero entries except for a $1$ in the $\ell$th position and $\ket{k}\bra{\ell}$ can be thought of as a matrix that is zero except for a $1$ in the $(k,\ell)$ position. $\ket{k}$ denotes a single excitation state with the excitation localized at the $k$th spin.

The Hamiltonian $H=H_{\rm single}$ of the system determines the time evolution of pure states $\ket{\Psi_0}$ via $\ket{\Psi(t)}=U(t)\ket{\Psi_0}$, where $U(t)$ is a one-parameter group of unitary operators governed by the Schr\"odinger equation 
\begin{equation}\label{e:SE}
  i\hbar \tfrac{d}{dt}U(t) = H U(t), \quad U(0)=I,
\end{equation}
where $I$ is the identity operator and $\hbar$ is the reduced Planck constant \new{(see, e.g.,~\cite[Eq. (1)]{survey_by_Rabitz}).}
By choosing time in units of $J^{-1}$ and energy in units of $J \times 1.05457173 \times 10^{-34} \joule \cdot \second$, we get $\hbar=J=1$ and can drop $\hbar$ in the following.

\subsection{Actuators for Spin Networks \& Control Paradigms}

Formally, an actuator for a quantum system is a device that interacts with the system, thereby altering its Hamiltonian---replacing $H$ by $H_S + H_C$. $H_S$ is the original system Hamiltonian, describing the intrinsic dynamics of the network, such as Eq.~\eqref{eq:HS} for the single excitation subspace. $H_C$ is a perturbation to the system Hamiltonian induced by the actuators, which can be constant or time-dependent. In the usual dynamic control framework for quantum systems, $H_C$ consists of one or more fixed interaction Hamiltonians $\mathsf{H}_{\mathsf{m}}$ with interaction strengths $\mathsf{u}_{\mathsf{m}}(t)$ that can be dynamically varied 
as 
\begin{equation}\label{eq:HC-dyn}
  H_C = \sum_{\mathsf{m}=1}^\mathsf{M} \mathsf{u}_{\mathsf{m}}(t) \mathsf{H}_{\mathsf{m}}.
\end{equation}
This results in a bilinear control problem for the controls $\mathsf{u}_{\mathsf{m}}(t)$.   A considerable amount of work on quantum control has focused on this paradigm  of time-dependent bilinear control~\cite{Dong_Petersen_survey, DAlessandro2003}.  This has proven to be a powerful tool and has been applied to controlling spin networks by dynamically varying all or some of the couplings $J_{k\ell}$ or potentials $D_k$~\cite{PRA2009}.

Usually finding suitable controls $\mathsf{u}_{\mathsf{m}}(t)$ is regarded as an open-loop control problem, but it can also be formulated in terms of finding a \emph{Feedback Control Law} (FCL),
\begin{equation}\label{eq:fb1}
  \mathsf{u}_{\mathsf{m}}(t) = \mathsf{u}_{\mathsf{m}} (U(t)).
\end{equation}
It is worth noting the differences between a FCL as defined above and Measurement-based Feedback Control (MFC) or Coherent Feedback Control (CFC) for quantum systems. FCLs such as Eq.\eqref{eq:fb1} are sometimes referred to as model-based feedback as the feedback is dependent on the evolution operator $U(t)$ of the system, which cannot be measured directly. Moreover, any measurement to obtain information about the evolution or current state of the system has a backaction that disturbs the system and thus acts as a co-actuator. In MFC, the state of the system is therefore usually replaced by an estimated state, which represents our state of knowledge about the system. It is obtained by state estimation based on continuous weak measurements. Incorporation of the measurement backaction and the probabilistic nature of quantum measurements further leads to stochastic differential equations and non-unitary evolution. CFC is another paradigm for quantum feedback based on coherent interaction between system and controller. This implicitly assumes that both, the system to be controlled and the controller, are quantum systems.   See~\cite{Altafini2012,quantum_meas_book,quantum_theory_methods} for good introductions to quantum control from a control engineering perspective.

All of these control paradigms play important roles in quantum control and are necessary to solve different problems~\cite{Schirmer2008}. MFC, for instance, is an important tool for deterministic state reduction and initial state preparation~\cite{Handel2005}. CFC can be used to stabilize quantum networks against noise and external perturbations~\cite{Zhang2012}. Dynamic open-loop control has found many applications from the preparation of quantum states of special interest, such as entangled states, and implementation of quantum gates for quantum information processing, to the control of spin dynamics in nuclear magnetic resonance (NMR), electron spin resonance (ESR), magnetic resonance imaging (MRI), and electronic, vibrational and rotational states of atoms and molecules~\cite{Glaser2015}. All of these paradigms, however, also have limitations and drawbacks. Dynamic control, for example, requires the ability to temporally modulate interactions, often at significant speed and time resolution. Besides, for networks with a high degree of symmetry such as rings with uniform coupling, controllability is often limited by dynamic symmetries, which imposes restrictions on what can be achieved, especially with local actuators~\cite{TAC2012}.

\begin{figure}
\centering\includegraphics[width=\columnwidth]{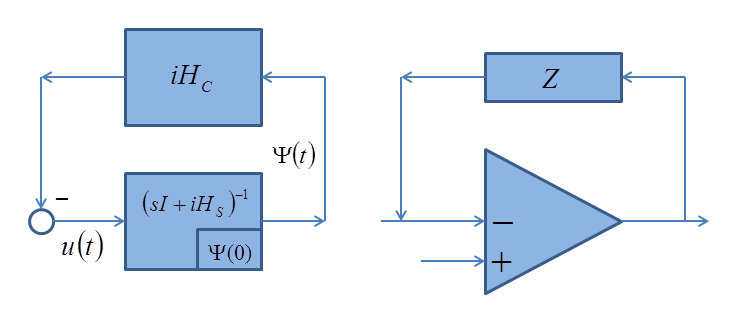}
\caption{Schematic of direct feedback loop for a quantum network (left) and conventional operational amplifier (right)}
\label{fig:schematic}
\end{figure}

Here we focus on the paradigm of finding constant interaction strengths as an alternative to dynamic control.  Specifically, we wish to design simple FCL's, 
\begin{equation}\label{eq:lfs0}
\vec{u}(U(t)) = -iH_C U(t)
\end{equation}
with time-invariant $H_C$, giving rise to a linear control system
\begin{equation}\label{eq:lfs}
  \tfrac{d}{dt}U(t) = (-iH_S) U(t) + I\vec{u}(U(t)),
\end{equation}
where $I$ is the identity matrix.  $H_S$ and $U(t)$ are complex operators but we could transform the system into a real system. \new{Decomposing the Schr\"odinger equation~\eqref{e:SE} as~\eqref{eq:lfs0}-\eqref{eq:lfs} and interpreting the control term as a feedback control law is not only conceptually but practically useful as it brings control insights to the problem.   
The addition of the control term creates a no-measurement ``direct feedback loop,'' 
a concept reminiscent of the seminal work of Bode~\cite{Bode}, where even though feedback exists no measurements are needed.   
Fig.~\ref{fig:schematic} attempts to illustrate the quantum control-feedback amplifier metaphor. }

\new{Dynamic control problems have been formulated in terms of model-based feedback and techniques such as Lyapunov control have been successfully applied to these problems, e.g.,~\cite{Wang2010,Altafini2007}, and even dynamic open-loop control schemes can be reformulated as time-varying FCLs.  However, our aim here is to find constant FCLs for certain tasks, while at the same time restricting the Hamiltonian to have a simple form.  Restricting the control of a bilinear system such as Eqs.~\eqref{e:SE}-\eqref{eq:HC-dyn} to be time-invariant reduces the design 
to a linear, but unconventional, control design~\cite{bilinear_constant_input}.} 

\section{\label{sec:design}Design of Optimal Feedback Control Laws for Excitation Transport}

\subsection{\label{s:control_tasks}Control Objectives}

Our main control objective is to transfer an initial state $\ket{\In}=\ket{m}$, corresponding to the initial excitation of the system on spin $m$, to a desired target state $\ket{\Out}=\ket{n}$, corresponding to the excitation on spin $n$, for any given pair $(m,n)$ of initial and target spins. Mathematically, we formulate the problem of arbitrary state transfer (not limited to single excitation states) as finding an \emph{input-output map} given by a unitary operator $U(T)$ that maximizes the (squared) \emph{fidelity} or \emph{probability of successful transfer} from $\ket{\In}$ to $\ket{\Out}$ in an amount of time $T$:
\begin{equation}\label{e:fidelity_T}
  p(\ket{\Out}\shortleftarrow \ket{\In},T) = |\bra{\Out}U(T)\ket{\In}|^2 \leq 1.
\end{equation}

In practice, readout of information is generally not instantaneous but takes place over a finite time window. In this case it is more advantageous to maximize the \emph{average transfer fidelity} for a given readout time window $2\delta T$,
\begin{equation}\label{e:fidelity_ave}
  \bar{p}(\ket{\Out}\shortleftarrow \ket{\In},T;\delta T) = \frac{1}{2\delta T}\int_{T-\delta T}^{T+\delta T} |\bra{\Out} U(t) \ket{\In}|^2 \;dt.
\end{equation}
Setting $\ket{\Out}=\ket{\In}$ and choosing a large readout time window $2\delta T$ we can suppress transport from time $0$ to $T$ and \emph{localize} or \emph{freeze excitations} at a particular node for later use by maximizing $\bar{p}(\ket{\In}\shortleftarrow\ket{\In},\tfrac{1}{2}T;\tfrac{1}{2}T)$.

Unitarity of $U(T)$ ensures selectivity of the transfer as $\norm{U(T) (\ket{\In}-\ket{\In'})} = \norm{\ket{\In}-\ket{\In'}}$, i.e., if $U(T)$ maps the input state to the target state then no other state can be mapped to the target state.  Quantitatively, an initial preparation error maps to a terminal error of the same magnitude  as that of the initial error.   The flipside of this selectivity requirement is that we cannot hope to engineer a process that renders the target state asymptotically stable but can only expect Lyapunov stability or Anderson localization~\cite{Anderson-58,50_years}.

\begin{figure}
\centering\includegraphics[width=.69\columnwidth]{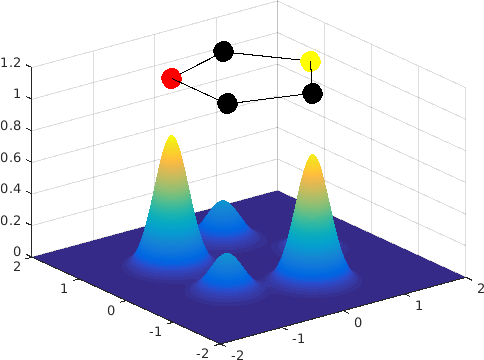}
\caption{Spin ring with energy landscape created by localized potentials.}
\label{fig:landscape}
\end{figure}

We are interested in control of information transfer by \emph{shaping} the potential \emph{energy landscape} of the system (see Fig.~\ref{fig:landscape}).  \new{The extent to which the energy landscape can be controlled in an actual device is subject to constraints, the precise nature of which depends on the physical realization.  However, there is generally some freedom to shape the energy landscape.  For example, there are proposals for semiconductor architectures consisting of quantum dots with surface gates that control the energy levels via the Stark shift.  In other architectures, magnetic fields (Zeeman shift) can be used to locally or globally control the energy landscape.  In atom traps, control of the energy landscape can be achieved by deforming the optical lattice~\cite{2016_Heisenberg_simulator}.  As this paper is mainly concerned with the development of a theoretical framework, details of experimental realizations and constraints are beyond the scope of the current work and are left for future work.}

Controlling the energy landscape means we wish to find a FCL $\vec{u}(U(t)) =  -i \widehat{D} U(t)$ with $\widehat{D}=\diag(\vec{D})$ and $\vec{D}=(D_1,\dotsc,D_N)$ that maximizes the probability of information transfer given by Eq.~\eqref{e:fidelity_T} or~\eqref{e:fidelity_ave}. For a network with fixed topology defined by the couplings $J_{k\ell}$, this corresponds to applying local potentials $D_k$ that are constant in time, resulting in a constant Hamiltonian $H_\vec{D} = H_S + \widehat{D}$ and an input-output map $U_{\vec{D}}(t)$ that is the solution of the Schr\"odinger Eq.~\eqref{e:SE} with $H=H_\vec{D}$. The objective is to find a control parameter vector $\vec{D}^*$ that maximizes the instantaneous transfer fidelity
\begin{equation}\label{eq:opt0}
  p_{\vec{D}^*}(\ket{\Out}\shortleftarrow \ket{\In},T) = \max_{\vec{D}} p_{\vec{D}} (\ket{\Out}\shortleftarrow \ket{\In},T)
\end{equation}
or the average transfer fidelity
\begin{equation}\label{eq:opt1}
  \bar{p}_{\vec{D}^*}(\ket{\Out}\shortleftarrow \ket{\In},T;\delta T) = \max_{\vec{D}} \bar{p}(\ket{\Out}\shortleftarrow \ket{\In},T;\delta T)
\end{equation}
at some time $T$. We can fix $T$, require $T \le T_{\max}$ with an upper bound $T_{\max}$, or aim to achieve the transfer with maximum fidelity in minimum time. We also wish to consider the \emph{sensitivity} of the transfer with regard to uncertainties in system parameters such as coupling strengths $J_{k\ell}$ and local potentials $D_k$ as well as disturbances such as environmental noise.

For practical applications, it is often preferable to modify the objective slightly and aim to find a FCL that achieves a desired transfer in \emph{minimum time} with a certain \emph{margin of error}, as we do not necessarily require perfect state transfer but only that the final state be sufficiently close, up to a global phase factor, to the desired target state. Finding FCLs for information transfer in spin networks thus reduces to an optimization problem, which can be solved using standard optimization tools. However, the optimization landscape is very challenging, in particular when the goal is to find a control that achieves the highest possible fidelity in the shortest time possible, possibly subject to various other constraints.

\subsection{\label{s:speed_limit}Perfect State Transfer \& Speed Limits}

There are many open questions regarding the existence of FCLs that achieve \emph{perfect} state transfer, in \emph{finite} time or asymptotically, and the control resources required. \emph{Perfect state transfer} from state $\ket{\In}$ to state $\ket{\Out}$ at time $T$ requires realization of a $U(T)$ such that $|\bra{\Out} U(T) \ket{\In}| = 1$. For perfect state transfer we have $\ket{\Out} = e^{i\phi} U(T) \ket{\In}$ with a global phase factor $\phi$. Hence, if the fidelity reaches its upper bound, we need not have $\ket{\Out}=U(T)\ket{\In}$, but only $[\ket{\Out}]=[U(T)\ket{\In}]$, where $[.]$ denotes the equivalence class of a unit vector of $\mathbb{C}^N$ in the complex projective space $\mathbb{C}\mathbb{P}^{N-1}\cong S^{2N-1}/S^1$.

It is easy to see that perfect state transfer is \emph{always} possible between \emph{any} pair of states in time $T$ for any $T>0$ if there are no constraints on the control Hamiltonian $H_C$, as we can simply set $H_C = -H_S + \tfrac{\pi}{2T} i(\ket{n}\bra{m}-\ket{m}\bra{n})$. However, the existence of FCLs that achieve perfect state transfer when the actuators are constrained is not obvious. Furthermore, even if such FCLs exist, information transfer is usually subject to \emph{speed limits}. While it is nontrivial to derive speed limits for arbitrary quantum networks, we can derive lower bounds on the transfer time in certain cases, which can be used as performance indicators for the optimization.

If the distance between initial and target spin is $1$, we can reduce the network to a two-spin system with direct coupling by applying large biases to sites other than the input and output spins, yielding an effective two-spin Hamiltonian
\begin{equation}
  H_\vec{D} = \begin{pmatrix} D_1 & 1 \\ 1 & D_2 \end{pmatrix}.
\end{equation}
This system undergoes Rabi oscillations with the Rabi frequency $\Omega=\sqrt{(D_2-D_1)^2+4}$ and it can easily be shown that
\begin{equation}
  p(\ket{2}\shortleftarrow \ket{1},T)=\left(\tfrac{1}{2}\Omega\right)^{-2}\sin^2\left(\tfrac{1}{2}{\Omega}T\right).
\end{equation}
The maximum of $1$ is achieved for $T=\tfrac{\pi}{2}$, if and only if $\Omega=2$, or $D_1=D_2$.

Similarly, if the distance between input and output spin is $2$, the network can be reduced to a three-spin chain by quenching it and assuming zero-bias on the three remaining spins. In this case we can easily show that
\begin{equation}
  p(\ket{3}\shortleftarrow \ket{1},T)=\sin^4\left(\tfrac{1}{2}\sqrt{2} T\right).
\end{equation}
Here we have perfect state transfer for $T=\tfrac{\pi}{2}\sqrt{2}$.

More generally, we can derive speed limits by quenching rings to chains from the eigenstructure symmetries. If the distance between input and output spin in a ring with $N$ spins with uniform nearest neighbor couplings is $n-1$ and the biases satisfy $D_{n+1-k}=D_k$, then $H_\vec{D}$ commutes with the permutation $\sigma=[n,n-1,\dotsc,1]$ with corresponding permutation matrix $P=P^\dag$, i.e. $PH_\vec{D}P = H_\vec{D}$. Let $V \Lambda V^\dag$ be an eigendecomposition of $H_{\vec{D}}$ with eigenvectors $\v_k$ and eigenvalues $\lambda_k$. Then $V \Lambda V^\dag = PV \Lambda V^\dag P$ or $V=PV$, i.e., the first and last row of $V$ are the same, and
the \emph{tracking error} $\|\ket{n}-e^{i\phi}U(T)\ket{1}\|$ with global phase factor $\phi$ becomes
\begin{equation*}
  \sum_k \left|(\v_{k})_1\right|^2 \left|1 - e^{-i(t\lambda_k-\phi)}\right|^2
     = 4 \sum_k \left|(\v_{k})_1\right|^2  \sin^2\left(\tfrac{1}{2}(t \lambda_k-\phi)\right).
\end{equation*}
(see Eq.~\eqref{eq:err} derived in Section~\ref{sec:optimality}). This expression vanishes if $t\lambda_k -\phi$ is an integer multiple of $2\pi$. For a chain of length three with no bias, $\lambda_1=-\lambda_3=\sqrt{2}$ and $\lambda_2=0$, and we achieve perfect state transfer for $T=2\pi/\lambda_1=\tfrac{1}{2}\sqrt{2}\pi$, setting $\phi=0$.

\section{\label{sec:optimality}Eigenstructure and Symmetry}

The observations about the role of symmetries and the Hamiltonian eigenstructure motivate a careful analysis of the role the latter play in the design of FCLs for information transfer in spin networks.

\subsection{Eigenstructure}

Consider the eigendecomposition of the Hamiltonian
\begin{equation}
  H_S+H_C=H_{\vec{D}} = \sum_{k=1}^{\Ne} \lambda_k \Pi_k,
\end{equation}
where $\Pi_k$ is the projector on the $k$th eigenspace associated with the eigenvalue $\lambda_k$ and $\Ne$ is the number of distinct eigenspaces. The eigenvalues $\lambda_k$ are real as the Hamiltonian is Hermitian. Furthermore, as in our case $H_\vec{D}$ is a real symmetric matrix, the projectors $\Pi_k$ are also real symmetric. The associated input-output map is
$U_\vec{D}(T) = \sum_{k=1}^{\Ne} e^{-iT\lambda_k} \Pi_k$. For the objective of maximizing the transfer fidelity at time $T$, we have
\begin{equation}
  \begin{split}
    \sqrt{p_\vec{\vec{D}}(\ket{n}\shortleftarrow \ket{m},T)}
      & = \left| \sum_{k=1}^{\Ne} e^{-iT\lambda_k} \bra{n}\Pi_k\ket{m} \right|\\
      & = \left| \sum_{k \in \K} e^{-i(T\lambda_k-\phi)} \bra{n}\Pi_k\ket{m} \right| \\
      & \le \sum_{k\in\K} \left| e^{-i(T\lambda_k-\phi)} \bra{n}\Pi_k\ket{m} \right| \\
      & = \sum_{k\in\K} \left|\bra{n}\Pi_k\ket{m} \right|,
  \end{split}\label{eq:p:b}
\end{equation}
where $\K$ is the subset of the eigenspaces that have non-zero overlap with the input and output state, $\bra{n}\Pi_k\ket{m} \neq 0$, and $\phi$ is a global phase factor that does not affect the norm. This means the maximum is achieved if (but not only if)
\begin{itemize}
\item[(i)] the phases of the exponentials cancel the phases of the projections $\bra{n}\Pi_k\ket{m}$, up to a \emph{global} phase factor $e^{i\phi}$ that is absorbed by the absolute value, and
\item[(ii)] $\sum_{k\in\K}|\bra{n}\Pi_k\ket{m}|$ is maximized simultaneously as the phase assignment.
\end{itemize}
The transfer is perfect if the upper bound $\sum_{k \in \K}|\bra{n}\Pi_k\ket{m}|=1$ is attained, in which case we call the controller \emph{superoptimal}.

To prove that the preceding conditions are not only sufficient but necessary, we observe the following:
\begin{align}
  \|\ket{n}-e^{i \phi}U_\vec{D}(T)\ket{m}\|^2
    &= 2-2\operatorname{Re}\left(\bra{n}e^{i \phi}U_\vec{D}(T)\ket{m}\right)\\
    &= 2\left(1-\sum_{k \in \K} \bra{n}\Pi_k\ket{m}\cos(T\lambda_k-\phi)\right). \notag
\end{align}
This yields
\begin{theorem}\label{t:necessary_and_sufficient}
Necessary and sufficient conditions for superoptimality are
\begin{enumerate}
\item[(i)] the eigenprojections of $H$ satisfy $\sum_k |\bra{n}\Pi_k\ket{m}|=1$;
\item[(ii)] the eigenvalues are such that the $T\lambda_k-\phi$'s are even or odd multiples of $\pi$ depending on whether the $\bra{n}\Pi_k\ket{m}$'s are positive or negative, resp.
\end{enumerate}
\end{theorem}
\begin{corollary} \label{c:nonuniformness}
For any $\widehat{D}$-controller and any $\ket{m} \ne \ket{n}$ it is impossible for all $\bra{n}\Pi_k\ket{m}$, $k\in \K$, to have the same sign.
\end{corollary}
\begin{proof}
  As $\{\Pi_k\}$ is a resolution of identity and  $\ket{m}\perp\ket{n}$,
  \begin{equation}\label{eq:zerosum}
    \sum_{k\in\K} \bra{n}\Pi_k\ket{m} = \bra{n} \sum_{k=1}^{N_e}\Pi_k\ket{m}= \ip{n}{m}=0.
  \end{equation}
\end{proof}

In the special case of $\K$ containing two elements, e.g., $\K=\{1,2\}$, Eq.~\eqref{eq:zerosum} yields $\bra{n}\Pi_1\ket{m}=-\bra{n}\Pi_2\ket{m}$ for any controller $\widehat{D}$ with the remaining states being dark, i.e., $\bra{n}\Pi_k\ket{m}=0$, $k\not\in\K$. The resulting freedom could be used to secure the phase condition along with $|\bra{n}\Pi_1\ket{m}|=|\bra{n}\Pi_2\ket{m}|=\tfrac{1}{2}$,
which yields $p_\vec{D}(\ket{n}\shortleftarrow \ket{m},T)=1$, i.e., perfect state transfer.

\subsection{Signature Property in the Case of Distinct Eigenvalues}

In the generic case when $H_\vec{D}$ has $N$ distinct eigenvalues, we have $\Pi_k = \ket{\vec{v}_k}\bra{\vec{v}_k}$, where $\{\v_k:k=1,\dotsc,N\}$ is the (real) orthonormal frame of eigenvectors of $H_\vec{D}$. Taking $(\v_{k})_m=\ip{\vec{v}_k}{m}$ and $(\v_{k})_n= \ip{\vec{v}_k}{n}$ to be the projections of the (real) input and output states onto the $k$th eigenvector of $H_\vec{D}$, the tracking error becomes
\begin{equation}\label{eq:err}
  \|\ket{n}-e^{i\phi}U_\vec{D}(T)\ket{m}\|^2 = \sum_k |(\v_{k})_n - e^{-i(t \lambda_k-\phi)} (\v_{k})_m|^2.
\end{equation}
It assumes its global minimum of $0$ if and only if
\begin{equation}
  \begin{array}{rl}
    |(\v_{k})_n| &= |(\v_{k})_m|, \\
    s_{kn}        &= \Sgn\left(e^{-i(t\lambda_k-\phi)}\right)s_{km},
  \end{array}
  \qquad \forall k = 1, \dotsc, N,
\end{equation}
where $s_{kn}:=\Sgn((\v_{k})_n)$ and $e^{-i(T\lambda_k-\phi)}$ is real at optimality. Noting $i e^{\pm i \pi/2} = \pm 1$, the previous condition is equivalent to
\begin{equation}\label{e:optimal_times}
  \begin{array}{rl}
    |(\v_{k})_n|           &= |(\v_{k})_m|, \\
    t\lambda_k - \phi &= \frac{\pi}{2}(s_{kn}-s_{km})  \!\!\! \mod 2\pi,
  \end{array} \qquad
  \forall k =1, \dotsc, N.
\end{equation}
Setting $s_k:=\Sgn((\v_{k})_n(\v_{k})_m)=\Sgn(\bra{n}\Pi_k\ket{m})$, we get
\begin{equation}\label{e:signature}
  (\v_{k})_n = s_k(\v_{k})_m.
\end{equation}
Even though only the $m$th and $n$th components of the eigenvectors matter in perfect state transfer, the signature property extends to other components related by symmetry.

\subsection{Symmetries \& Full Signature Property of Eigenvectors}\label{s:full_signature}

The key to finding good feedback control laws by optimization lies in understanding the \emph{symmetries} of the system and using the biases to enforce or annul certain symmetries. For this we constrain the controls to ensure that the first condition, $|(\v_{k})_n|=|(\v_{k})_m|$ is satisfied for all $k\in\K$ and for all admissible controls. Let $H_{\vec{D}}=V \Lambda V^\dag$ be an eigendecomposition of $H_\vec{D}$. If there is a unitary operator $R$ that commutes with $H_\vec{D}$, $RH_\vec{D}= H_\vec{D}R$, then $H_\vec{D}  = R H_\vec{D} R^\dag = R V \Lambda V^\dag R^\dag = R V \Lambda (RV)^\dag$ implies that if $\vec{v}_k$ is a unit eigenvector with eigenvalue $\lambda_k$ then so is $R\vec{v}_k$. If the eigenvalues $\lambda_k$ of $H_\vec{D}$ are distinct, then both vectors can only differ by a phase, $R\vec{v}_k=e^{i\phi}\vec{v}_k$; in particular
\begin{equation}
  |\ip{n}{\vec{v}_k}| = |\ip{n}{R\vec{v}_k}| = |\ip{R^\dag n }{\vec{v}_k}|,  \qquad \forall k\in \K.
\end{equation}
Hence, we need to find a unitary operator $R$ that commutes with $H_\vec{D}$ and satisfies $R^\dag\ket{n}=\ket{m}$.

\begin{example}
For a chain of length $N$ with uniform coupling we have inversion symmetry, i.e., the system Hamiltonian $H_S$ commutes with the permutation operator $P$, $H_SP=PH_S$, where $P\ket{n}=\ket{N+1-n}$ for $n=1,\dotsc,N$. If the control Hamiltonian $H_C=\widehat{D}$ also commutes with $P$, then $|\left(\v_k\right)_m|=|\left(\v_k\right)_n|$ for all $k$ whenever the input and output node satisfy $m=N+1-n$.
\end{example}
\begin{example}
For a ring of $N$ spins with uniform coupling, we have translation invariance in addition to inversion symmetry. Therefore, we can always choose biases such that $|\left(\v_k\right)_m|=|\left(\v_k\right)_n|$ for all $k$.
\end{example}

For a ring with uniform coupling we can show that Eq.~\eqref{e:signature} not only holds for the $(m,n)$ input-output components, but also for those components related by the permutation $\sigma(m+j)=n-j$.   Motivated by the cyclic symmetry of the ring requiring modulo $N$ operations, we relabel the indices of the spins, starting at $0$ rather than $1$ and the indices are taken modulo $N$ without indicating this explicitly to keep the notation simple. By convention, the labeling of the vertices is clockwise around the ring.

\begin{theorem}
For a ring of $N$ spins with uniform coupling between adjacent spins only, the eigenvectors are signature symmetric as
\begin{equation}\label{e:v-symmetry}
  \left(\v_k\right)_{m+\ell} = s_k \left(\v_k \right)_{n-\ell}, \quad  m+\ell, n-\ell \in\{0,1,\dotsc,N-1\}
\end{equation}
under the symmetry of the biases
\begin{equation}
   \label{e:bias-symmetry}
  D_{m+\ell}=D_{n-\ell}, \quad  m+\ell, n-\ell \in\{0,1,\dotsc,N-1\}.
\end{equation}
Furthermore, if $|m-n|$ is even then $(\v_k)_{ m+\ell=n-\ell}=0$.
\end{theorem}

The proof is given in Appendix~\ref{app:signature-property}.

\section{\label{sec:sensitivity}Sensitivity to Uncertainties}

\subsection{General Sensitivity}

We analyze the sensitivity of the (squared) fidelity or probability of successful transfer $\left|\bra{n}U(T)\ket{m}\right|^2$ relative to uncertainties in the couplings $J_{k\ell}$ or other parameters. Let
\begin{equation}\label{eq:Hpert}
  \Hp = H_S + H_C + \sum_{\mu} \delta_\mu S_\mu
\end{equation}
be the total Hamiltonian of the perturbed system, where $H_S$ is the ideal system Hamiltonian, $H_C$ is the control Hamiltonian, here assumed to be time-invariant, and the $\delta_\mu S_\mu$ for $\mu=1,2,\dotsc$ are perturbations. $S_\mu$ reflects the \emph{structure} of the perturbation and $\delta_\mu$ its amplitude. For uncertainty in the coupling $J_{k\ell}$, we take $S_\mu=\ket{k}\bra{\ell} + \ket{\ell}\bra{k}$. The transfer operator of the perturbed system is
\begin{equation}\label{eq:Upert}
  \Up(T) = e^{-i(H_S+H_C+\sum_\mu \delta_\mu S_\mu)T}.
\end{equation}
The design sensitivity is determined by the partial derivative
\begin{equation}\label{eq:partial}
  \frac{\partial |\bra{n}\Up(T)\ket{m}|^2}{\partial \delta_\mu}
    = 2 \operatorname{Re}\left( \bra{n}\frac{\partial \Up(T)}{\partial \delta_\mu}\ket{m}\bra{m}\Up^\dag(T)\ket{n}\right).
\end{equation}
As $H_S+H_C$ generally does not commute with the perturbation $S_\mu$, we use the general formula to evaluate the partial derivative
\begin{equation}
  \frac{\partial \Up(T)}{\partial \delta_\mu} = -i\int_0^1 e^{-i\Hp T(1-s)}(S_\mu T)e^{-i \Hp Ts} \; ds,
\end{equation}
which remains valid around $\delta_\mu \neq 0$ (see~\cite{neat_formula,krotov_neat_formula}). From the eigendecomposition of the perturbed Hamiltonian $\Hp$, $\Hp = \sum_{k=1}^{\Np} \lambdap_k \Pip_k$, where $\Pip_k$ is the projector onto the eigenspace associated with the eigenvalue $\lambdap_k$, it is readily found that
\begin{equation*}
  \bra{m}\Up(T)^\dag\ket{n}  = \sum_j \bra{m}\Pip_j \ket{n} e^{i T\lambdap_j}.
\end{equation*}
Evaluation of the integral gives
\begin{align}
  \frac{\partial \Up(T)}{\partial \delta_\mu}
    &= -iT \sum_{k,\ell} \Pip_k S_\mu \Pip_\ell \int_0^1e^{-i T\lambdap_k(1-s)}e^{-i T\lambdap_\ell s}\;ds\nonumber\\
    &= -iT \sum_{k,\ell} \Pip_k S_\mu \Pip_\ell \frac{e^{-i T\lambdap_\ell}-e^{-i T\lambdap_k}}{i T\left(\lambdap_k-\lambdap_\ell\right)}.\label{e:integral}
\end{align}
Inserting this into Eq.~(\ref{eq:partial}) gives
\begin{align}\label{e:old_sensitivity}
    & \frac{\partial |\bra{n}\Up(T)\ket{m}|^2}{\partial \delta_\mu} \nonumber\\
  = & \; 2T\sum_{j,k,\ell} \bra{m} \Pip_j \ket{n} \bra{n} \Pip_k S_\mu \Pip_\ell \ket{m}  \nonumber\\
    & \qquad \times \frac{\cos\left(T\left(\lambdap_k-\lambdap_j\right)\right)
        -\cos\left(T\left(\lambdap_\ell-\lambdap_j\right)\right)}{T\left(\lambdap_k-\lambdap_\ell\right)}\\
  = & -2T \sum_{j,k,\ell} \bra{m} \Pip_j \ket{n} \bra{n} \Pip_k S_\mu \Pip_\ell \ket{m} \nonumber\\
    & \qquad \times \frac{\sin\left(\tfrac{1}{2}T\left(\lambdap_k-\lambdap_\ell\right)\right)}{\tfrac{1}{2}T\left(\lambdap_k-\lambdap_\ell\right)}
  \sin\left(\tfrac{1}{2}T\left(\lambdap_k+\lambdap_\ell-2\lambdap_j\right)\right),\nonumber
\end{align}
where we used $\cos(a)-\cos(b) =-2\sin(\tfrac{1}{2}(a-b))\sin(\tfrac{1}{2}(a+b))$. Finally, defining $\omegap_{k\ell}=\lambdap_k-\lambdap_\ell$ and $\sin(x)/x =\sinc(x)$, gives
\begin{equation}\label{e:new_sensitivity}
  \begin{split}
  \frac{\partial |\bra{n}\Up(T)\ket{m}|^2}{\partial \delta_\mu}
    = & -2T\sum_{k,\ell}  \bra{n} \Pip_k S_\mu \Pip_\ell \ket{m} \sinc\left(\tfrac{1}{2}T\omegap_{k\ell}\right)\\
      & \times \sum_j \bra{m} \Pip_j \ket{n} \sin\left(\tfrac{1}{2}T(\omegap_{kj}+\omegap_{\ell j})\right).
  \end{split}
\end{equation}

\subsection{Sensitivity at $\delta_\mu=0$}

Up to now, the sensitivity could have been evaluated at any $\delta_\mu$. From here on, we restrict the discussion to $\delta_\mu=0$. One of the implications of Theorem~\ref{t:necessary_and_sufficient}, saying that at superoptimality $T \lambda_k$ is a multiple of $\pi$ modulo the global phase factor $\phi$, is that for $k=\ell$ the argument of the sine in Eq.~\eqref{e:new_sensitivity} is a multiple of $\pi$ (the global phase factors in $\tfrac{1}{2}T(\lambda_k+\lambda_\ell-2\lambda_j)$ cancel) and thus the sine vanishes. Therefore, the sum over $k,\ell$ in Eq.~\eqref{e:new_sensitivity} can be restricted to $k \ne \ell$. Next, observe that
\begin{equation*}
   \sin\left(\tfrac{1}{2}T(\lambda_k+\lambda_\ell-2\lambda_j)\right) = s_j\sin\left(\tfrac{1}{2}T(\lambda_k+\lambda_\ell)\right).
\end{equation*}
This allows us to isolate the sum over $j$, which takes the value $1$ at superoptimality:
\begin{equation*}
  \sum_j s_j \bra{m}\Pi_j\ket{m}=\sum_j |\bra{m}\Pi_j\ket{n}|=1.
  \end{equation*}
Finally, observe that $\sinc\left(\tfrac{1}{2}T(\lambda_k-\lambda_\ell)\right)$ vanishes when $T(\lambda_k-\lambda_\ell)$
is a multiple of $2\pi$. Putting everything together, the sensitivity formula becomes
\begin{align}
  \frac{\partial |\bra{n}\Up(T)\ket{m}|^2}{\partial \delta_\mu}
    =& -2T\sum_{s_ks_\ell=-1}  \bra{n} \Pip_k S_\mu \Pip_\ell \ket{m} \sinc\left(\tfrac{1}{2}T\omegap_{k\ell}\right) \notag \\
     &  \times \sin\left(\tfrac{1}{2}T(\lambda_k+\lambda_\ell)\right), \label{e:latest_sensitivity}
\end{align}
where $s_ks_\ell=-1$ indicates that the sum is restricted to those $k,\ell$ such that $T(\lambda_k-\lambda_\ell)$ is an odd multiple of $\pi$.

\begin{figure*}
  \includegraphics[width=\textwidth]{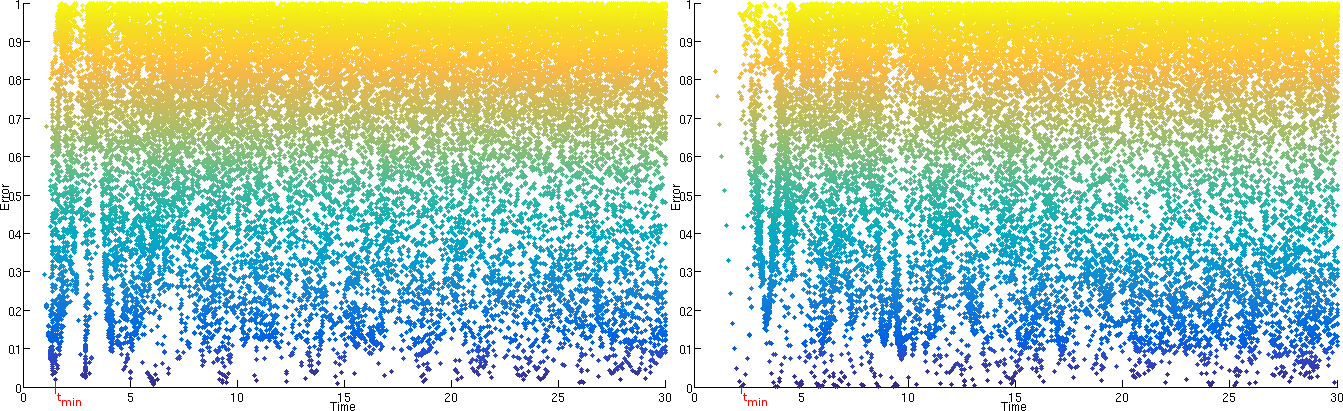}
  \caption{Results of optimizing the information propagation from spin $1$ to $2$ (left) and $1$ to $3$ (right) for an XX-ring of $9$ spins using L-BFGS optimization with exact gradients over spatial biases for fixed times $T$ from $1$ to $30$ with step size $0.2$. Each data point represents the infidelity $1-p(\ket{\Out}\shortleftarrow\ket{\In},T)$ achieved for the corresponding time by a single optimization run across different initial values with $100$ restarts for each time $T$. The optimization gets trapped often but it can still find good solutions for certain times. In particular, good solutions were found for the speed limit $t_{\mbox{min}}$ for both transfers.}\label{fig:ring9-errors}
\end{figure*}

\subsection{Vanishing Sensitivity to Symmetric Perturbations at Optimality}

We prove that the sensitivity of the fidelity relative to a real perturbation structured as $S_\mu=S_\mu^\dagger$ vanishes. This includes a perturbation of the $\mu\text{-}(\mu+1)$ coupling, in which case $S_\mu=\ket{\mu}\bra{\mu+1}+\ket{\mu+1}\bra{\mu}$, and the case where the perturbation is on the $D_\mu$ control bias, in which case $S_\mu=\ket{\mu}\bra{\mu}$. The vanishing of the latter sensitivity is quite trivial; indeed, if the controller is differentiably optimal, the first order conditions require that the directional derivative of the fidelity along any control direction, in particular $\ket{\mu}\bra{\mu}$, must vanish.

By a real Gram-Schmidt orthonormalization process, we write $\Pi_k = \ket{\v_k}\bra{\v_k}$, so that Eq.~\eqref{e:latest_sensitivity} can be rewritten,
\begin{align}\label{e:sensitivity_relabeled}
    \frac{\partial |\bra{n}\Up(T)\ket{m}|^2}{\partial \delta_\mu}
      =& -2T\sum_{s_ks_\ell=-1}  \ip{n}{\v_k}\ip{\v_\ell}{m} \qf{\v_k}{S_\mu}{\v_\ell}\\
       &  \times \sinc(\tfrac{1}{2}T(\lambda_k-\lambda_\ell)\sin(\tfrac{1}{2}T(\lambda_k+\lambda_\ell)). \notag
\end{align}
Observe that $\sinc(\tfrac{1}{2}T(\lambda_k-\lambda_\ell))\sin(\tfrac{1}{2}T(\lambda_k+\lambda_\ell))$ is symmetric relative to the indices $k,\ell$. That is,
\begin{eqnarray*}
  \lefteqn{\sinc(\tfrac{1}{2}T(\lambda_k-\lambda_\ell))\sin(\tfrac{1}{2}T(\lambda_k+\lambda_\ell))}\\
    &=&\sinc(\tfrac{1}{2}T(\lambda_\ell-\lambda_k))\sin(\tfrac{1}{2}T(\lambda_\ell+\lambda_k)).
\end{eqnarray*}
Likewise, $\qf{\v_k}{S_\mu}{\v_\ell}$ is symmetric relative to the indices $k,\ell$, $\qf{\v_k}{S_\mu}{\v_\ell}=\qf{\v_\ell}{S_\mu}{\v_k}$, because $S_\mu$ is a real symmetric matrix and the eigenvectors were taken to be real.

If on the right-hand side of Eq.~\eqref{e:sensitivity_relabeled} we add the same right-hand side with interchanged indices $k$ and $\ell$, we obtain twice the partial derivative of the squared fidelity relative to $\delta_\mu$. Thus,
\begin{equation}\label{e:sensitivity_final}
  \begin{split}
    2&\frac{\partial |\bra{n}\Up(T)\ket{m}|^2}{\partial \delta_\mu}\\
     &= -2T\sum_{s_ks_\ell=-1} \qf{\v_k}{S_\mu}{\v_\ell} (\ip{n}{\v_k}\ip{\v_\ell}{m}+\ip{n}{\v_\ell}\ip{\v_k}{m}) \\
     & \qquad \times \sinc(\tfrac{1}{2}T(\lambda_k-\lambda_\ell))\sin(\tfrac{1}{2}T(\lambda_k+\lambda_\ell)).
  \end{split}
\end{equation}
Next, we use the signature property to derive the following:
\begin{align*}
  \lefteqn{\ip{n}{\v_k}\ip{\v_\ell}{m}+\ip{n}{\v_\ell}\ip{\v_k}{m}}\\
    &=s_k\ip{m}{\v_k}\ip{\v_\ell}{m}+s_\ell\ip{m}{\v_\ell}\ip{\v_k}{m}\\
    &=s_k(\ip{m}{\v_k}\ip{\v_\ell}{m}+s_\ell s_k\ip{m}{\v_\ell}\ip{\v_k}{m})\\
    &=s_k(\ip{m}{\v_k}\ip{\v_\ell}{m}-\ip{m}{\v_\ell}\ip{\v_k}{m})\\
    &=0.
\end{align*}
Therefore, the right-hand side of Eq.~\eqref{e:sensitivity_final} vanishes and the sensitivity vanishes.

\begin{theorem}\label{t:clash}
Consider a spin ring in its single excitation subspace with biases ${\bf D}=(D_1,\ldots,D_N)$ which differentiably maximize the fidelity. At optimality the sensitivity of the fidelity relative to any real, symmetric perturbation $S_\mu$ vanishes.
\end{theorem}

\section{\label{sec:optimization}Optimization in a Challenging Landscape}

We present numerical optimization results for instantaneous and average information transfer and localization as well as the corresponding sensitivities of the controllers versus their performance. Initial results on controlling instantaneous information flow in spin networks are available in~\cite{time_optimal}, which are summarized and expanded here. The results here are computed for uniform rings of $N$ spins with XX couplings. All results for $N$ from $3$ to $20$ are available in a separate data set~\cite{data_set}.

Solving the optimization problem in Eq.~\eqref{eq:opt0} directly for a fixed target time $T$ is challenging even without constraints on the biases $\vec{D}$ as the landscape is extremely complicated with many local extrema, resulting in trapping of local optimization approaches such as quasi-Newton methods. Fig.~\ref{fig:ring9-errors} shows the results of various runs for fixed times for a ring of $9$ spins, with the objective being to propagate the excitation from spin $1$ to spin $2$ and $3$, respectively. While good solutions are found for certain times, including the minimum times given by the quantum speed limits (see Section~\ref{s:speed_limit}), the optimization clearly gets trapped frequently, making finding good solutions very expensive.

Instead of fixing the transfer time we add the time as additional parameter to optimize over. The results in Section~\ref{sec:optimality} show that the structure of the eigenvalues and eigenvectors must fulfill a specific condition to be able to maximize the transfer fidelity. While there are many potential structures that fulfill the condition, we can choose a specific one to provide a guide for good initial values and a restricted domain for the search. The idea is to quench the ring of $N$ spins into a chain from the initial spin to the target spin. Previous work showed that this can be easily achieved by applying a very strong potential in the middle between initial and target spin~\cite{rings_QINP}. If we can control the potentials of all spins then we can generalize this to quench the ring just before the initial and after the target spin, giving two options for a chain connecting the two nodes where either could provide a solution.

Furthermore, applying mirror symmetric potentials across the axis through the middle between initial and target state in the ring gives rise to an eigenstructure satisfying the optimality conditions. Consequently, we choose such symmetric potentials in combination with the approximate times where the maximum fidelity is achieved in the related chains as initial values for the optimization. This significantly improves the efficiency of finding controls for maximum information transfer in minimum time, as already observed in~\cite{time_optimal}. The symmetry constraint can be easily applied to the optimization by reducing the number of biases to be found to $\lceil N/2\rceil$, symmetric across the symmetry axis between initial and target spin. Convergence of the optimization can be further improved by selecting constants, peaks or troughs as biases between initial and target spin on both sides of the rings randomly as initial values, and selecting initial times from the transition times required for a spin chain of length corresponding to the distance between $\ket{\In}$ and $\ket{\Out}$.

\begin{figure}
  \centering\includegraphics[width=\columnwidth,height=.3\textheight]{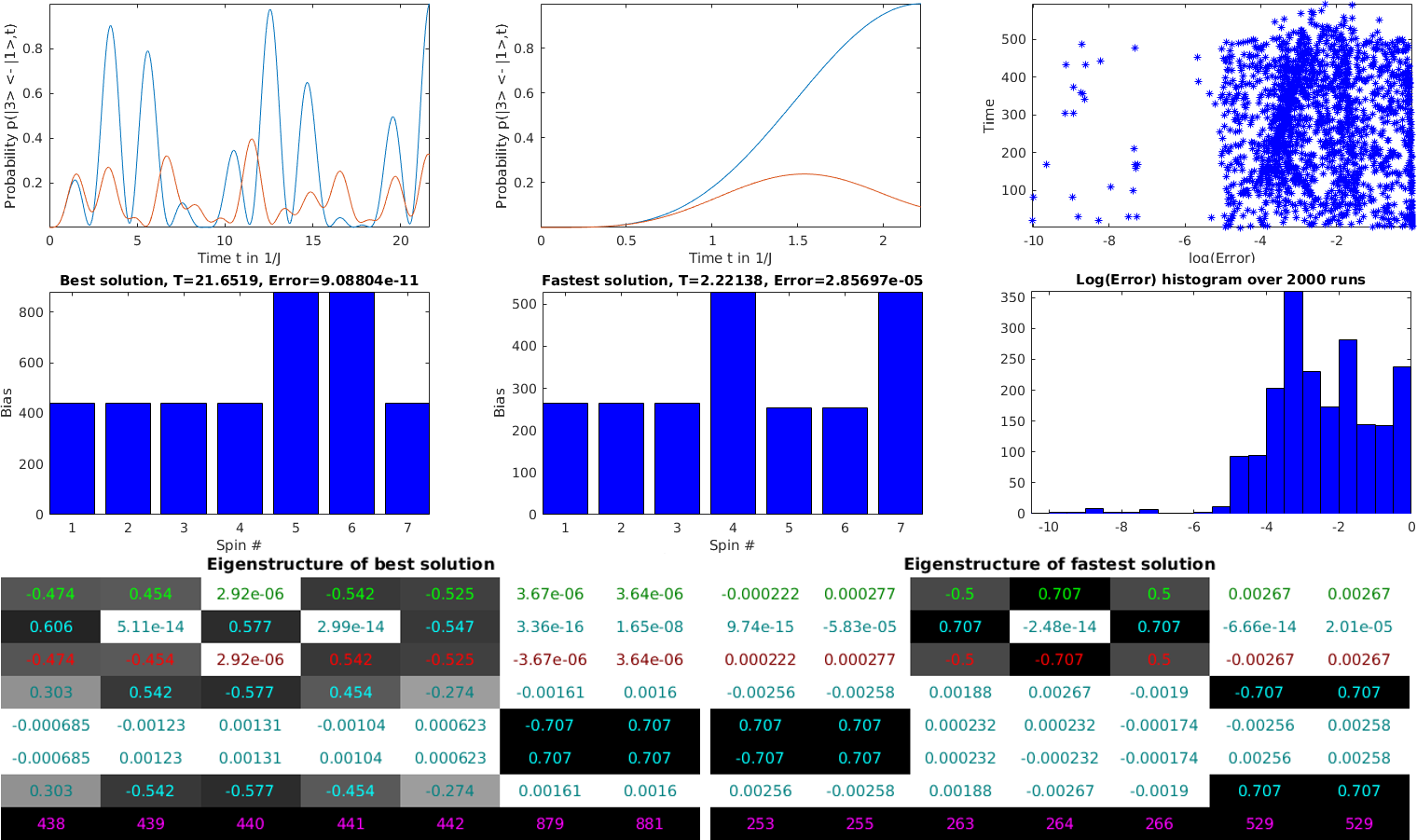}
  \caption{Optimization results for the information transfer probability from spin $1$ to $3$ for an XX-ring of $7$ spins over  spatial biases and time. The left column shows the biases and evolution (in blue vs. the natural evolution in red) giving the best fidelity at time $T \approx 21.65$ with an error of $9.09\times10^{-11}$. The middle column shows the fastest solution found with a fidelity greater than $0.999$ at $T \approx 2.22$. The right column shows the overall solutions found by repeated optimization, plotting time vs logarithm of the infidelity and a histogram of the logarithm of the infidelity. The bottom row shows the eigenstructure of the best and the fastest solution and their symmetries, with eigenvectors being the columns of the matrices (in cyan; green and red rows indicate $\ket{\In}$ and $\ket{\Out}$ states resp.) and corresponding eigenvalues at the bottom (in purple).}\label{fig:bias_t-7-3}
\end{figure}

Figs.~\ref{fig:bias_t-7-3}, \ref{fig:bias_t-11-3} show the optimization results for a ring of size $7$ and $11$ for the transition from spin $1$ to $3$. We report the solution with the highest fidelity and the fastest solution with a fidelity larger than $0.999$. Typically the highest fidelity solutions are found at longer times, but good solutions for short times are also achieved. However, many restarts of the optimization are required, and many runs fail with fidelities smaller than $0.9$. Observe the eigenstructure symmetries for the solutions consistent with Section~\ref{sec:optimality}. Shortest time solutions are found, while the best solution is at a different time.

We also report results for optimizing the average transfer fidelity, Eq.~\eqref{eq:opt1}: see Fig.~\ref{fig:bias_dt-11-6} for a $11$-ring for the transition from spin $1$ to $6$ and Fig.~\ref{fig:bias_dt-13-3} for a $13$ ring from spin $1$ to $3$. We show the solution with the highest fidelity and the fastest solution with a fidelity larger than $0.99$, lower than in the instantaneous case as the average fidelities are smaller as well.

Figs.~\ref{fig:fastest_t} and~\ref{fig:fastest_dt} show the shortest times achieved for instantaneous fidelities greater than $0.999$ and average fidelities greater than $0.99$ for rings of size $N=3$ to $20$ in summary. Due to the symmetry in the connections only transitions from $\ket{1}$ to $\ket{\lceil N\rceil/2}$ are reported. For target spins $\ket{2}$ and $\ket{3}$, the fastest times are generally consistent with the speed limits in Section~\ref{s:speed_limit}, but the shortest times could not always be achieved. All individual results can be accessed in a separately data set~\cite{data_set}. The cases where no minimum time solution satisfying the minimum fidelity requirements was found further show the difficulty of finding good controllers. Improved optimization strategies will be explored in future work.

Optimizing the average information transfer fidelity per Eq.~\eqref{e:fidelity_ave} can also be used to localize the excitation at a particular spin by maximizing $\bar{p}(\ket{\In}\shortleftarrow\ket{\In},\tfrac{1}{2}T;\tfrac{1}{2} T)$ as noted in Section~\ref{s:control_tasks}. Numerical results for rings of size $14$ and $19$ are shown in Figs.~\ref{fig:bias_dt-14-1} and~\ref{fig:bias_dt-19-1} for a holding time of $T=1,000$.

Theorem~\ref{t:clash} indicates that at superoptimality the sensitivity vanishes, which is further explored numerically here. Results for instantaneous transfer, time-window average transfer shown in Figs.~\ref{fig:sensitivity_t}, \ref{fig:sensitivity_dt}, \ref{fig:sensitivity_dt1} indicate a positive correlation between the sensitivity of the controllers and the infidelity. The specific sensitivity measure used here is the norm of the vector of sensitivities w.r.t. uncertainties in the $J_{k\ell}$ couplings. Among $2,000$ controllers indexed in decreasing order of fidelity, where we only show those with fidelities greater than $0.1$, the very best controllers nearly achieving the upper bound on the fidelity (achieving near vanishing tracking error) have nearly vanishing sensitivity; furthermore, with the deterioration of the fidelity the sensitivity increases.

\begin{figure}
  \centering\includegraphics[width=\columnwidth,height=.2\textheight]{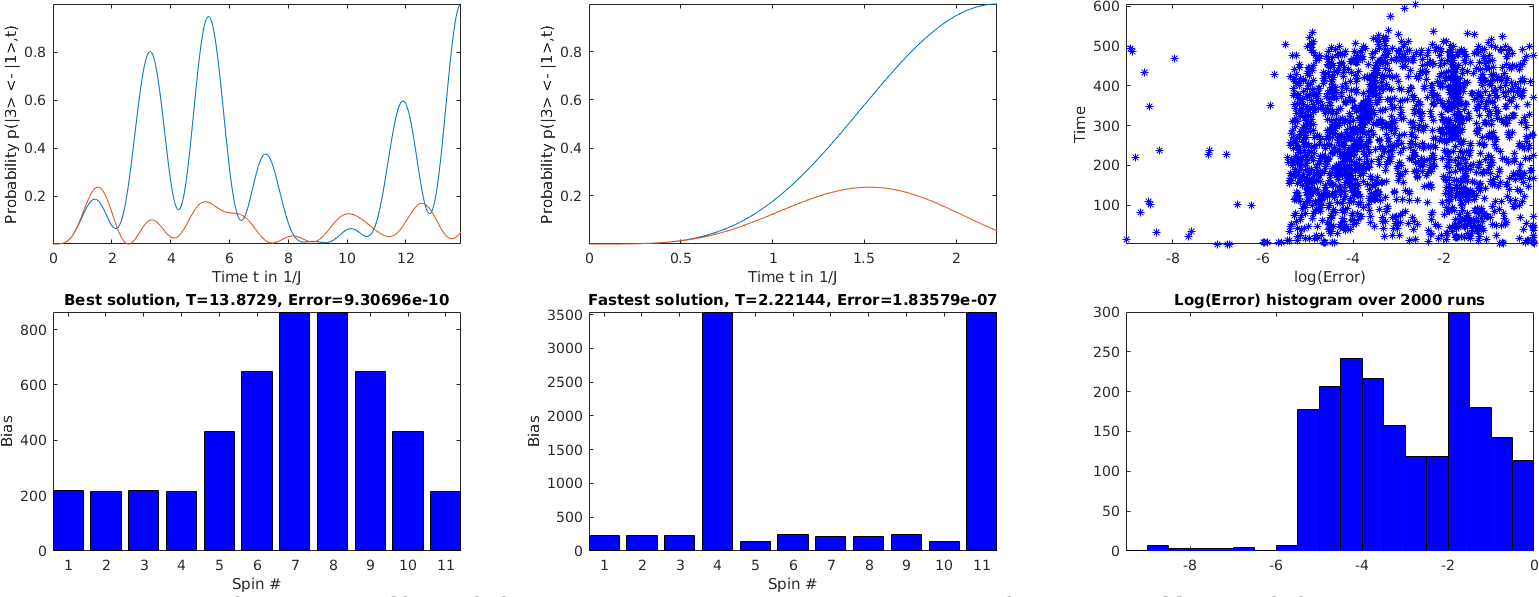}
  \caption{Results for optimizing the information transfer probability from spin $1$ to $3$ for a ring of $11$ spins similar to Fig.~\ref{fig:bias_t-7-3} (without eigenstructures).}
\label{fig:bias_t-11-3}
\end{figure}

\section{\label{sec:robust}Classical versus Quantum Robust Control}

Given a loop matrix $L$, a classical result is that the sensitivity $S=(I+L)^{-1}$ mapping from the reference to the tracking error and the logarithmic sensitivity of the sensitivity, $T=L(I+L)^{-1}$, 
derived from $S^{-1}(dS)=-(dL)L^{-1} T$, are in conflict since $S+T=I$. Horowitz~\cite[Chap. Six]{Horowitz} was probably first to point out that the limitation imposed by the SISO single degree-of-freedom configuration could be overcome by means of a two-degrees-of-freedom configuration. Ever since this fundamental observation, many MIMO two-degrees-of-freedom architectures have been proposed~\cite{2_deg_freedom_controller,Robust_2DOF_MIMO}. The controller $-i\widehat{D}(\ket{n},\ket{m}) \ket{\Psi}$ is, in a certain sense, a two-degrees-of-freedom controller as it depends on \emph{both}, the current state $\ket{\Psi(t)}$ and the target state $\ket{n}$, and does not explicitly depend on the tracking error. However, as already alluded to in Sec.~\ref{s:speed_limit}, the information transfer controller does not have a tracking error in the classical sense, but a projective tracking error.

\begin{figure}
  \centering\includegraphics[width=\columnwidth,height=.2\textheight]{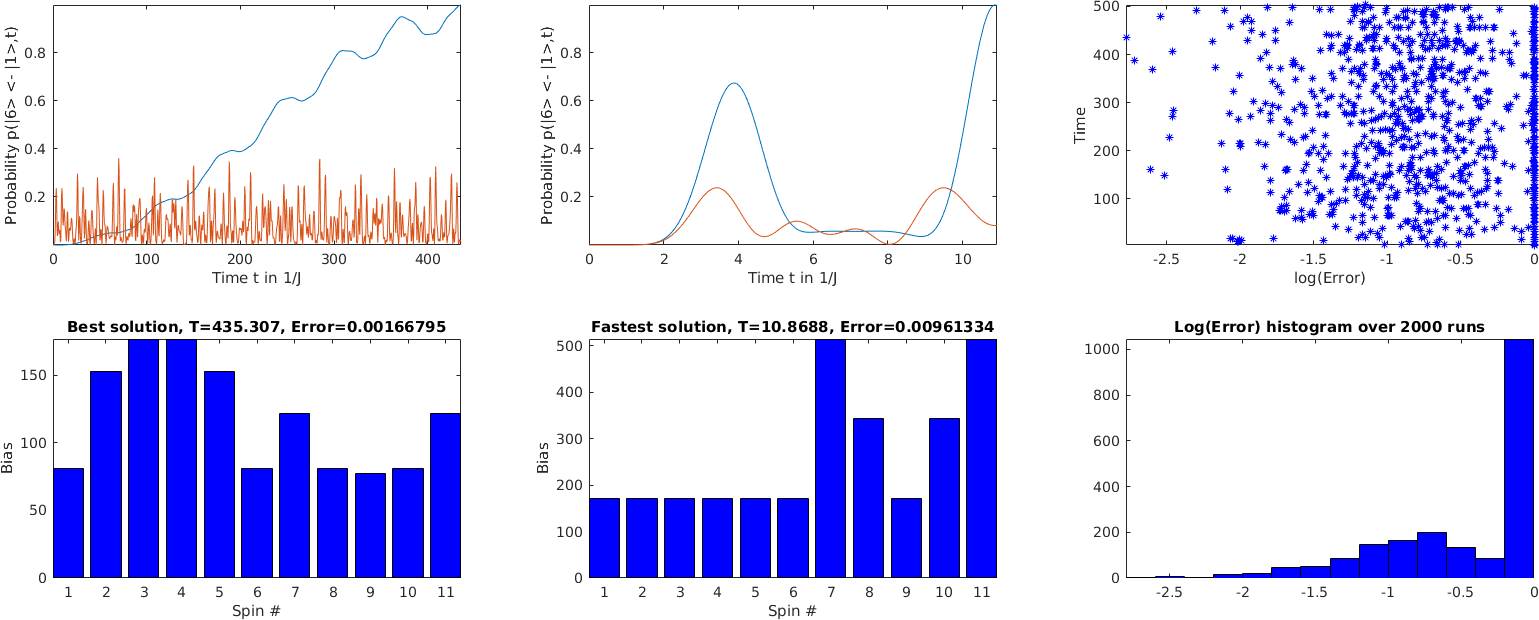}
  \caption{Results for optimizing the average information transfer probability from spin $1$ to $6$ for a ring of $11$ spins with $\delta T=0.05$ similar to Fig.~\ref{fig:bias_t-7-3} with the only difference that the fastest solutions with a fidelity of greater than $0.99$, due to the averaging, has been selected.}\label{fig:bias_dt-11-6}
\end{figure}

\begin{figure}
  \centering\includegraphics[width=\columnwidth,height=.2\textheight]{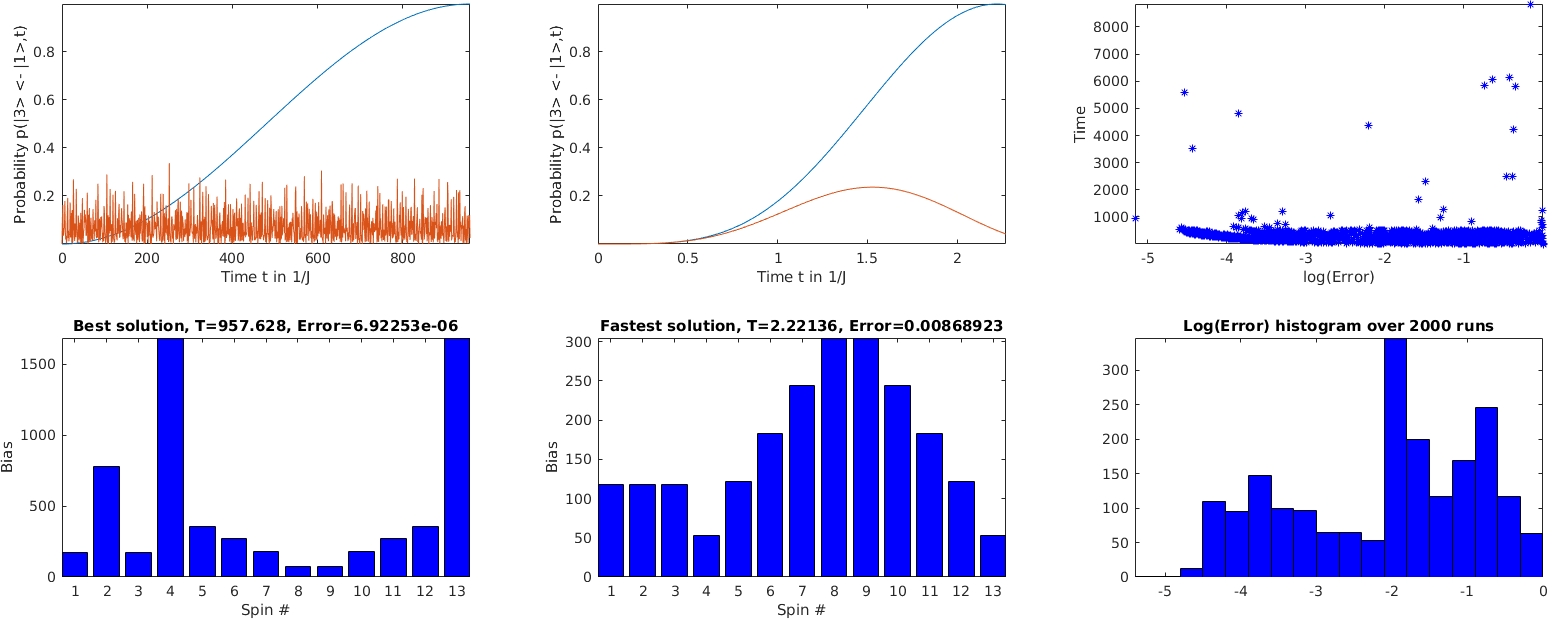}
  \caption{Results for optimizing the average information transfer probability from spin $1$ to $3$ for a ring of $13$ spins with $\delta T=0.05$ similar to Fig.~\ref{fig:bias_t-7-3} with the only difference that the fastest solutions with a fidelity of greater than $0.99$, due to the averaging, has been selected.}\label{fig:bias_dt-13-3}
\end{figure}

\begin{figure}
  \centering\includegraphics[width=\columnwidth,height=.25\textheight]{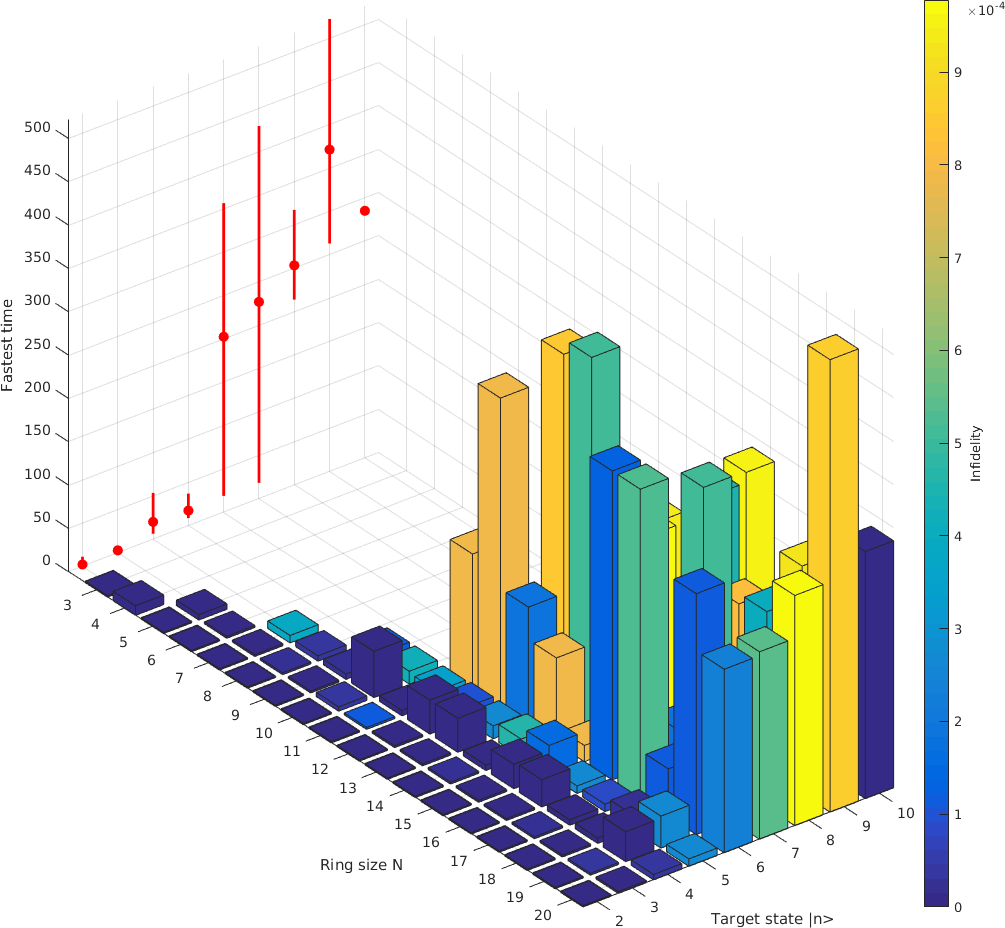}
  \caption{Shortest times achieved for instantaneous transition fidelities greater than $0.999$ for rings of size $N=3,\dots,20$ and transitions from $1$ to $k=2,\dotsc,\lceil N/2 \rceil$. Note that for the transitions for $N=13$ from $\ket{1}$ to $\ket{8}$, $N=16$ from $\ket{1}$ to $\ket{9}$ and $N=19$ from $\ket{1}$ to $\ket{10}$ no solution with fidelity greater than $0.999$ were found, so no fastest results are reported. The color of the bars indicate the infidelity of the fastest solution.}\label{fig:fastest_t}
\end{figure}

\begin{figure}
  \centering\includegraphics[width=\columnwidth,height=.25\textheight]{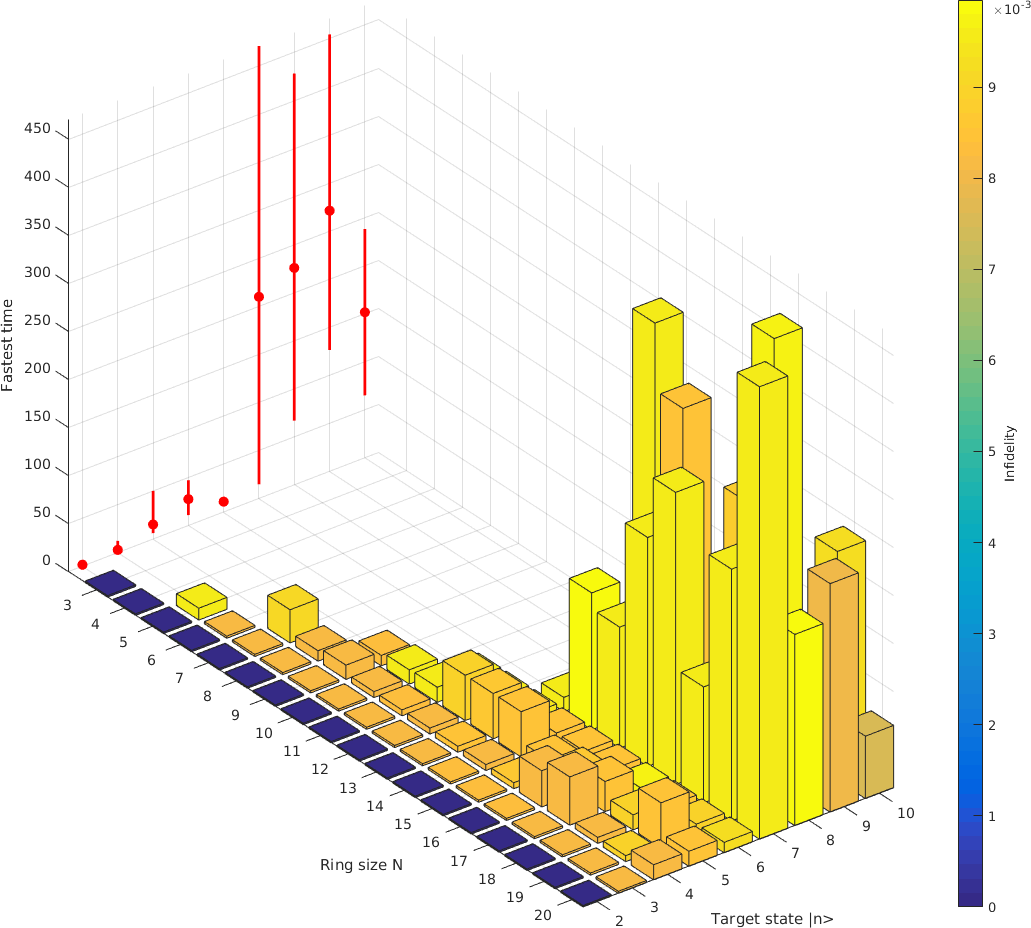}
  \caption{Shortest times achieved for average transition fidelities greater than $0.99$ for rings of size $N=3,\dots,20$ and transitions from $1$ to $k=2,\dotsc,\lceil N/2 \rceil$ for $\delta T = 0.05$. Note that for the transition for $N=19$ from $\ket{1}$ to $\ket{8}$ no solution with fidelity greater than $0.99$ were found, so no fastest results are reported. The color of the bars indicate the infidelity of the fastest solution.}\label{fig:fastest_dt}
\end{figure}

\begin{figure}
  \centering\includegraphics[width=\columnwidth,height=.2\textheight]{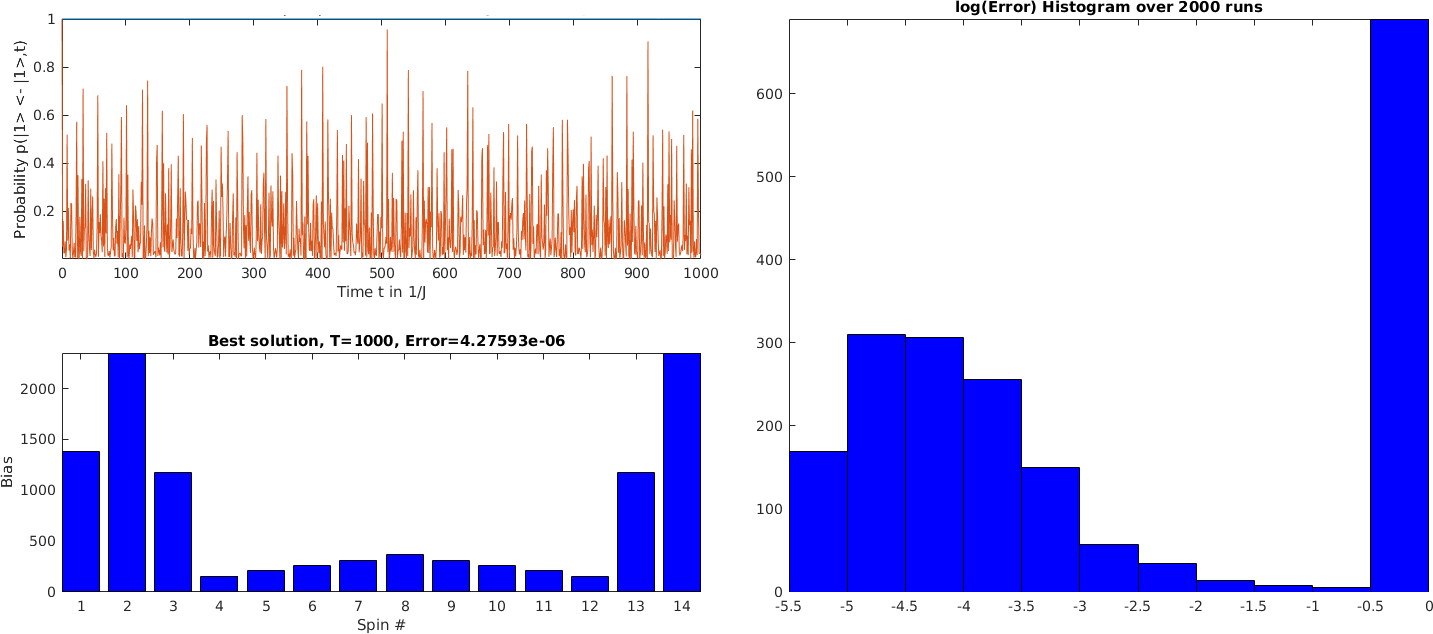}
  \caption{Optimization results for localizing spin $1$ in a $14$-ring over the spatial biases. The left column shows the biases and evolution (in blue vs. the natural evolution in red) giving the best fidelity for a localization time of $1,000$ with an error of $4.28\times10^{-6}$.}\label{fig:bias_dt-14-1}
\end{figure}

\begin{figure}
  \centering\includegraphics[width=\columnwidth,height=.2\textheight]{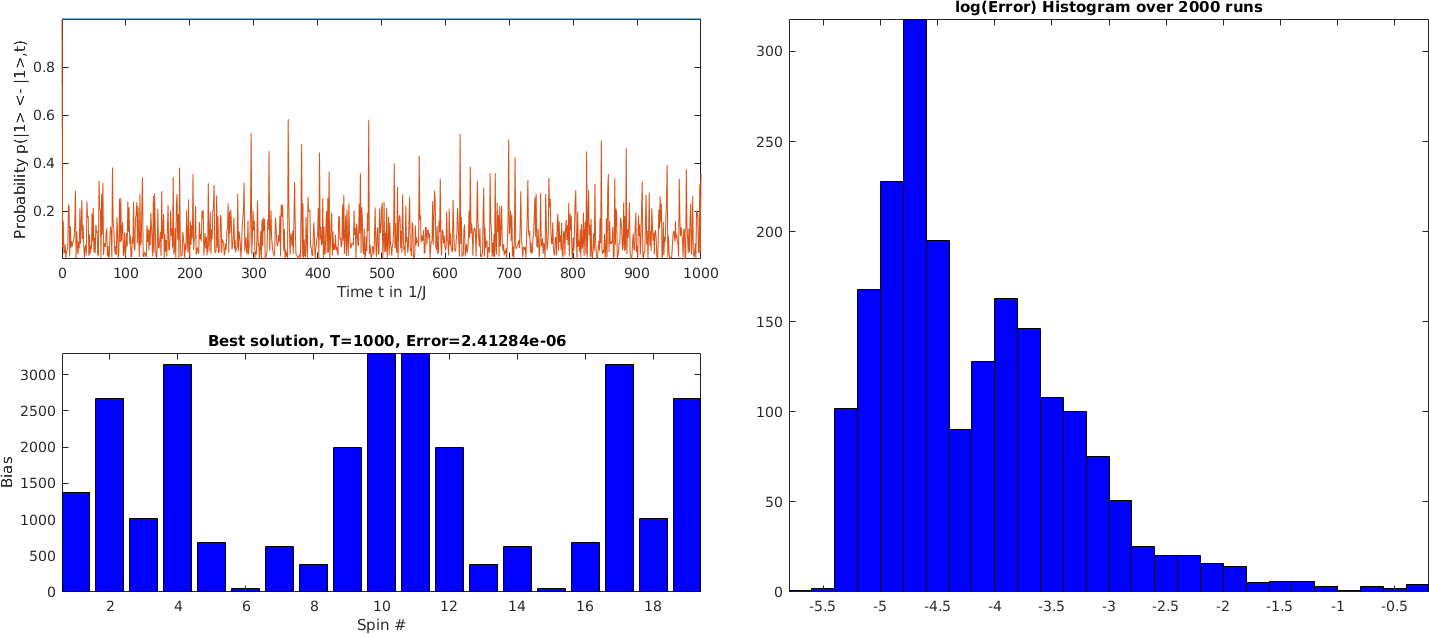}
  \caption{Optimization results for localizing spin $1$ in a $19$-ring similar to  Fig.~\ref{fig:bias_dt-14-1}.}\label{fig:bias_dt-19-1}
\end{figure}

To proceed towards classical Laplace domain control, consider the quantum mechanical projective tracking error
\[ E(t)=\ket{n}1(t)-e^{i \phi(t)}\ket{\Psi(t)} =\left(\ket{n}-e^{i \phi(t)}e^{-i (H_S+\widehat{D})t} \ket{m}\right)1(t), \]
where $1(t)$ denotes the unit step. The phase factor $e^{i \phi(t)}$ is a generalization of the phase factor of Section~\ref{s:speed_limit} securing
\begin{equation}\label{e:error_vs_overlap}
  \|E(t)\|^2=2-2|\ip{n }{\Psi(t)}|,
\end{equation}
that is, minimization of $\|E(t)\|$ is equivalent to maximization of $|\ip{n}{\Psi(t)}|$. It is easily seen that the phase factor to secure the above equality is $\phi(t)=-\phase{\ip{n}{\Psi(t)}}$. This creates an unconventional (adaptive) feedback from $\Psi$ to $\phi$. Instead of minimizing $\|E(t)\|^2$ or maximizing $|\ip{n}{\Psi(t)}|$ over $(t,D)$ at a specific time, we could optimize in a time-average sense, opening the road to Laplace transform techniques.

The Laplace transform of the error reads
\[ \mathcal{L}[E(t)](s)=\underbrace{\left( \frac{1}{s}I-\mathcal{L}\left[e^{i \phi(t)}\right]\ast \left(sI+i \left(H_S+\widehat{D}\right)\right)^{-1}P\right)}_{=:\mathcal{S}(s)}\ket{n}, \]
where $P$ is a permutation matrix such that $\ket{m}=P\ket{n}$ and $\ast$ denotes the Laplace domain convolution. Since $\mathcal{S}(s)$ is the mapping from the unit step reference to the error, it can be interpreted as a sensitivity matrix, but it differs significantly from the classical sensitivity matrix. In particular, selectivity implies that only the $n$th column of $\mathcal{S}(s)$ matters.

\begin{figure}
  \includegraphics[width=.49\columnwidth,height=.28\textheight]{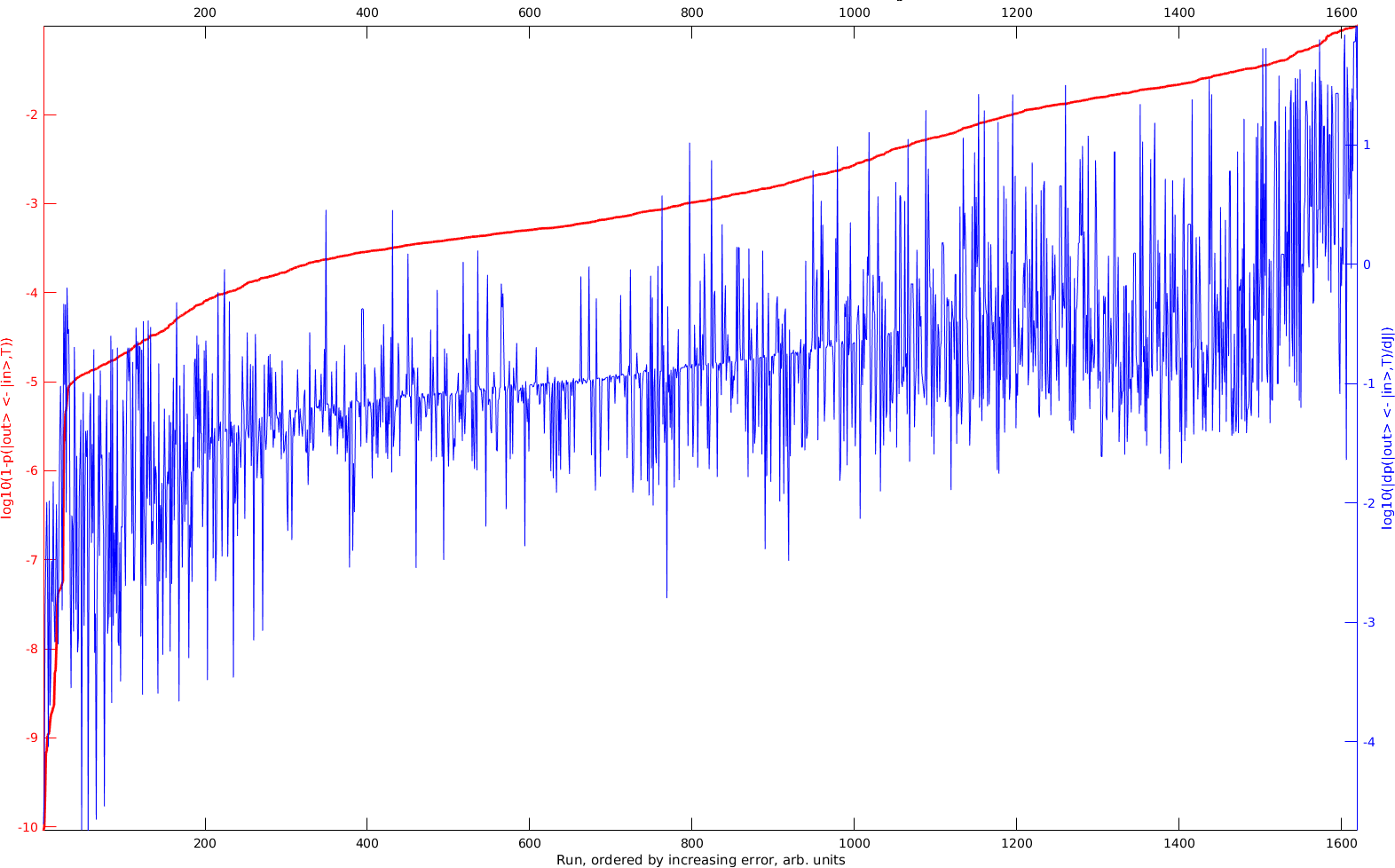}\hfill
  \includegraphics[width=.49\columnwidth,height=.28\textheight]{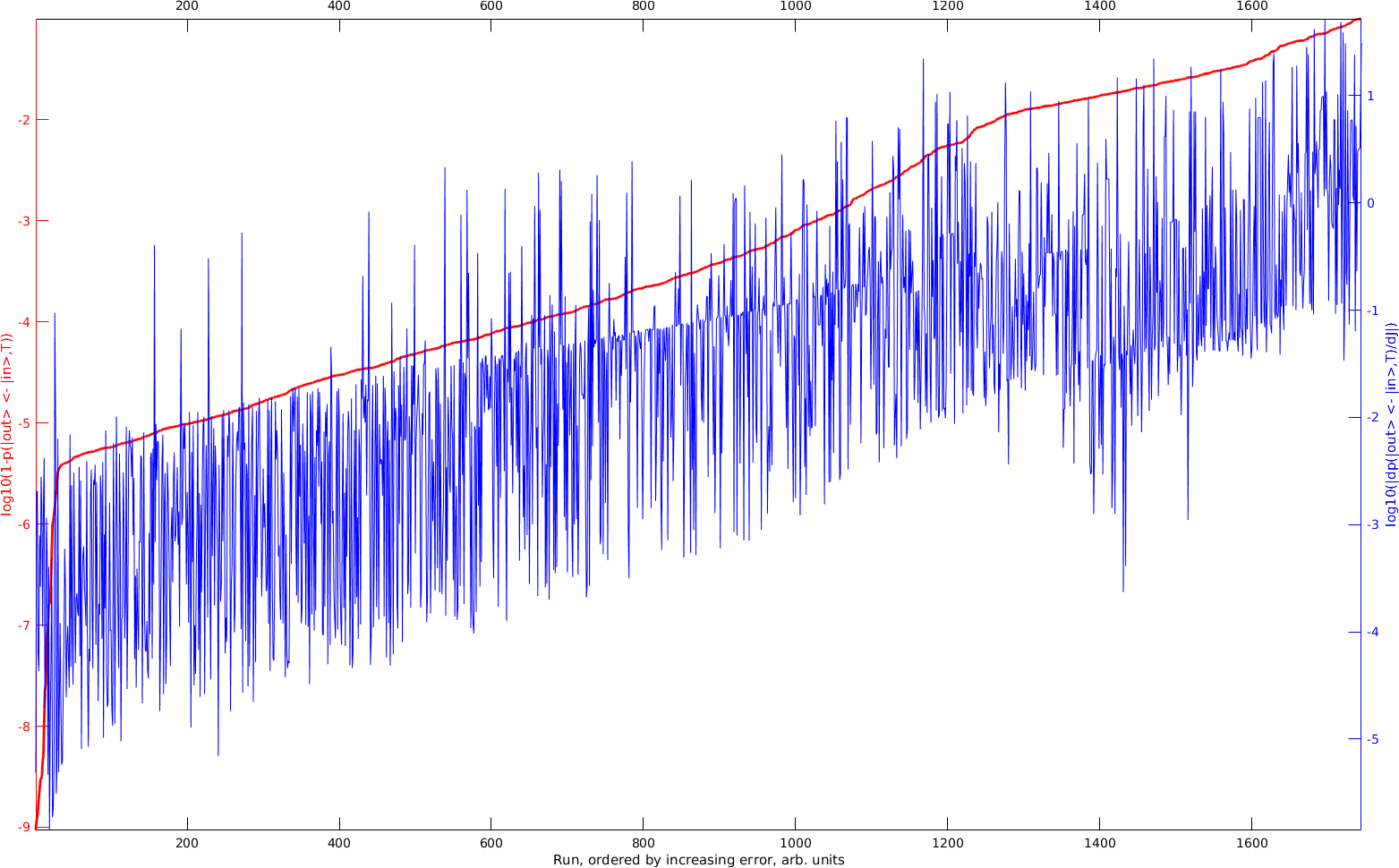}
  \caption{Logarithm of infidelity $1-p$ (red) and logarithm of sensitivity (blue), ordered by increasing infidelity from left to right, of the instantaneous $1 \to 3$ controllers of a $7$-ring (left) and a $11$-ring (right).}\label{fig:sensitivity_t}
\end{figure}

\begin{figure}
  \includegraphics[width=.49\columnwidth,height=.28\textheight]{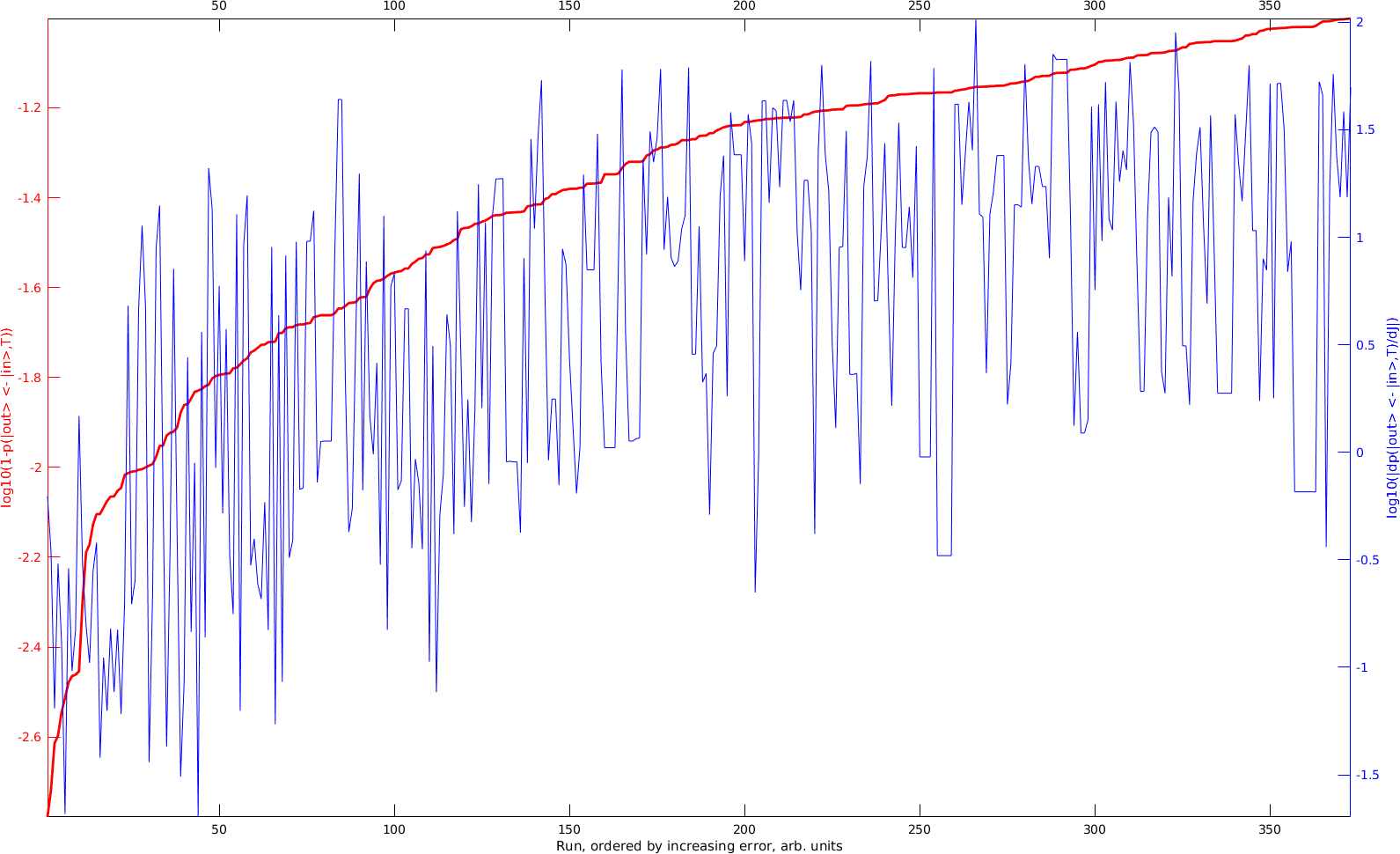}\hfill
  \includegraphics[width=.49\columnwidth,height=.28\textheight]{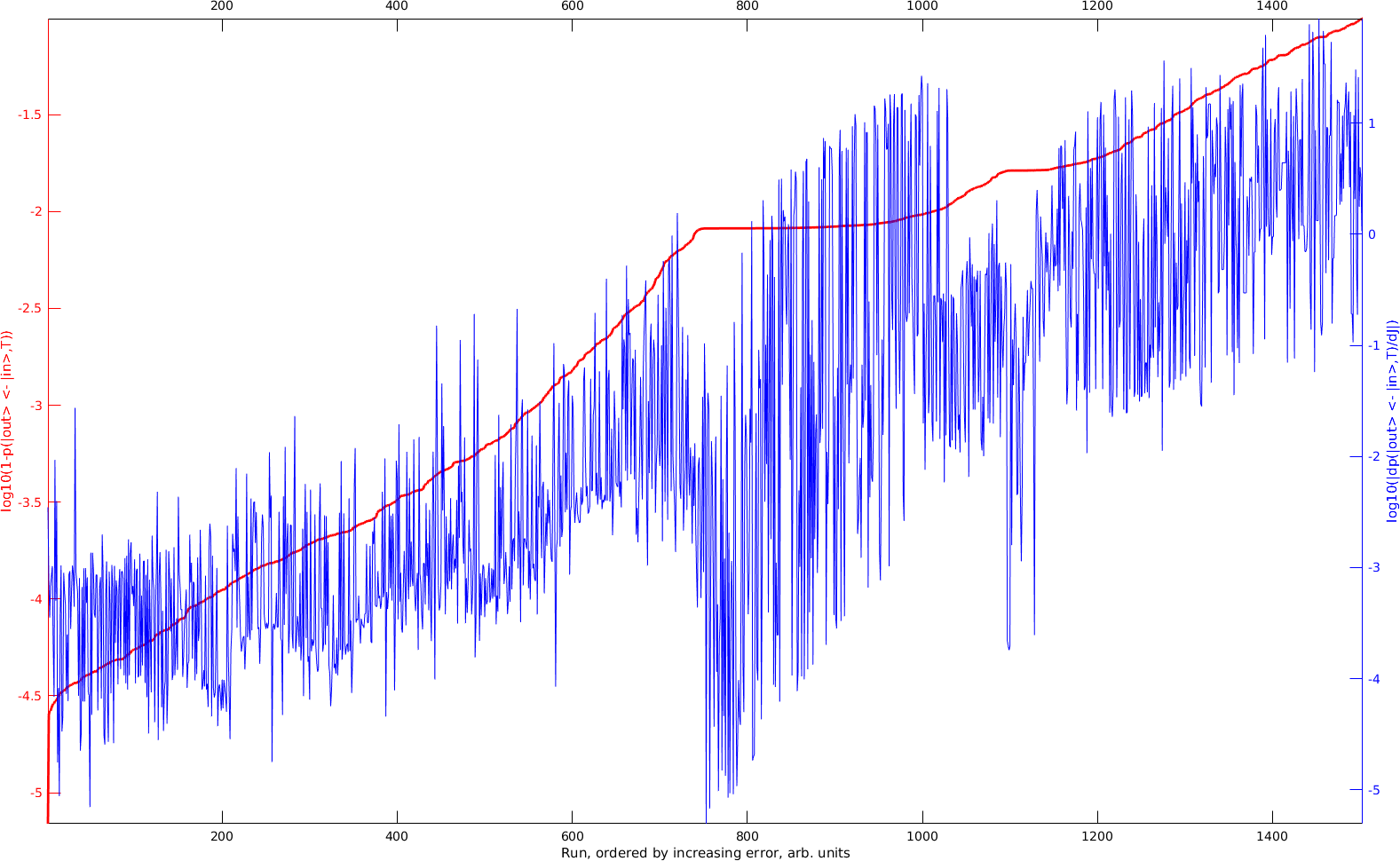}
  \caption{Logarithm of infidelity $1-\bar{p}$ (red) and logarithm of sensitivity (blue), ordered by increasing infidelity from left to right, of the average $1 \to 6$ controllers of a $11$-ring (left) and $1 \to 3$ controllers of an $11$-ring (right).}\label{fig:sensitivity_dt}
\end{figure}

\begin{figure}
  \includegraphics[width=.49\columnwidth,height=.28\textheight]{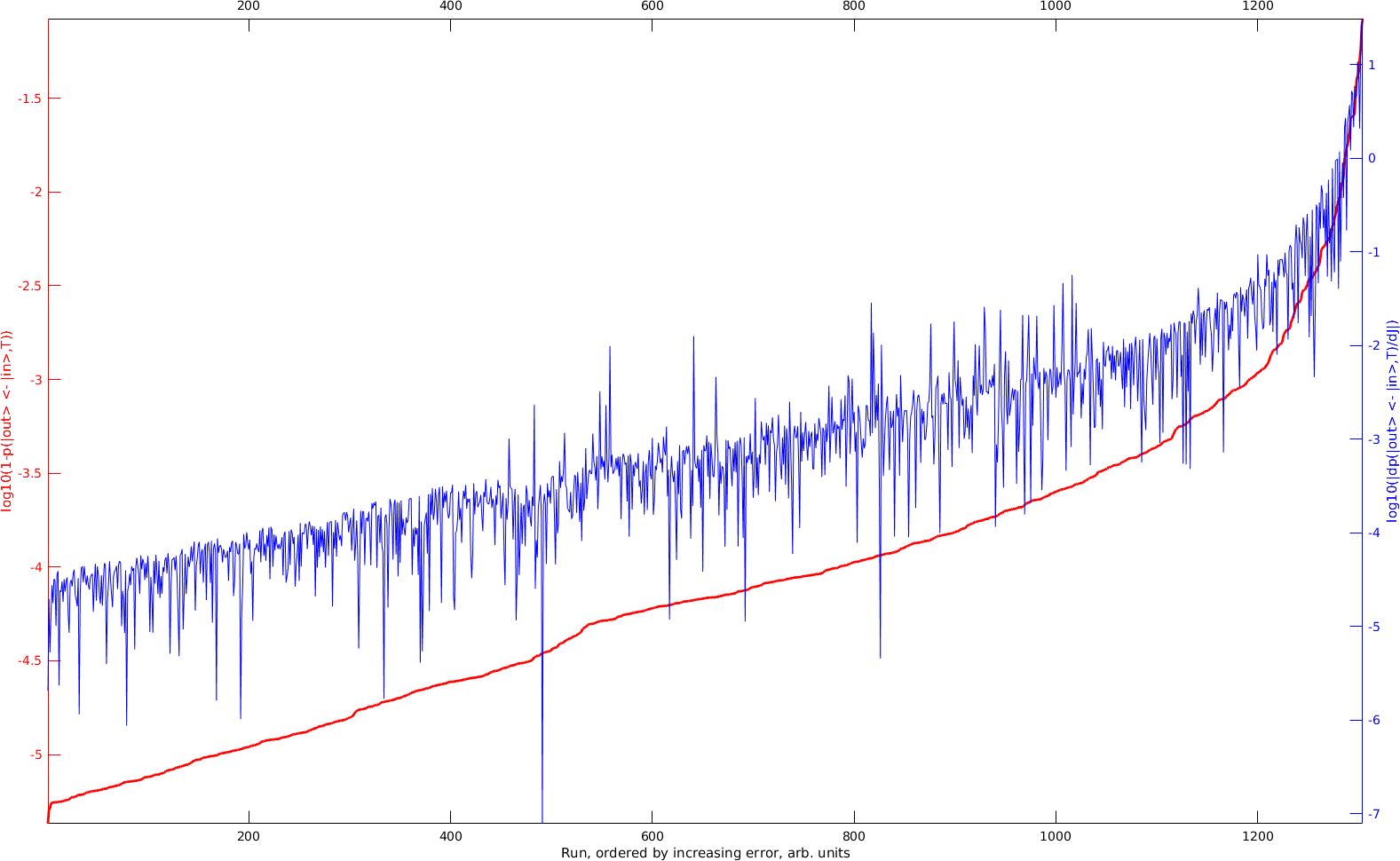}\hfill
  \includegraphics[width=.49\columnwidth,height=.28\textheight]{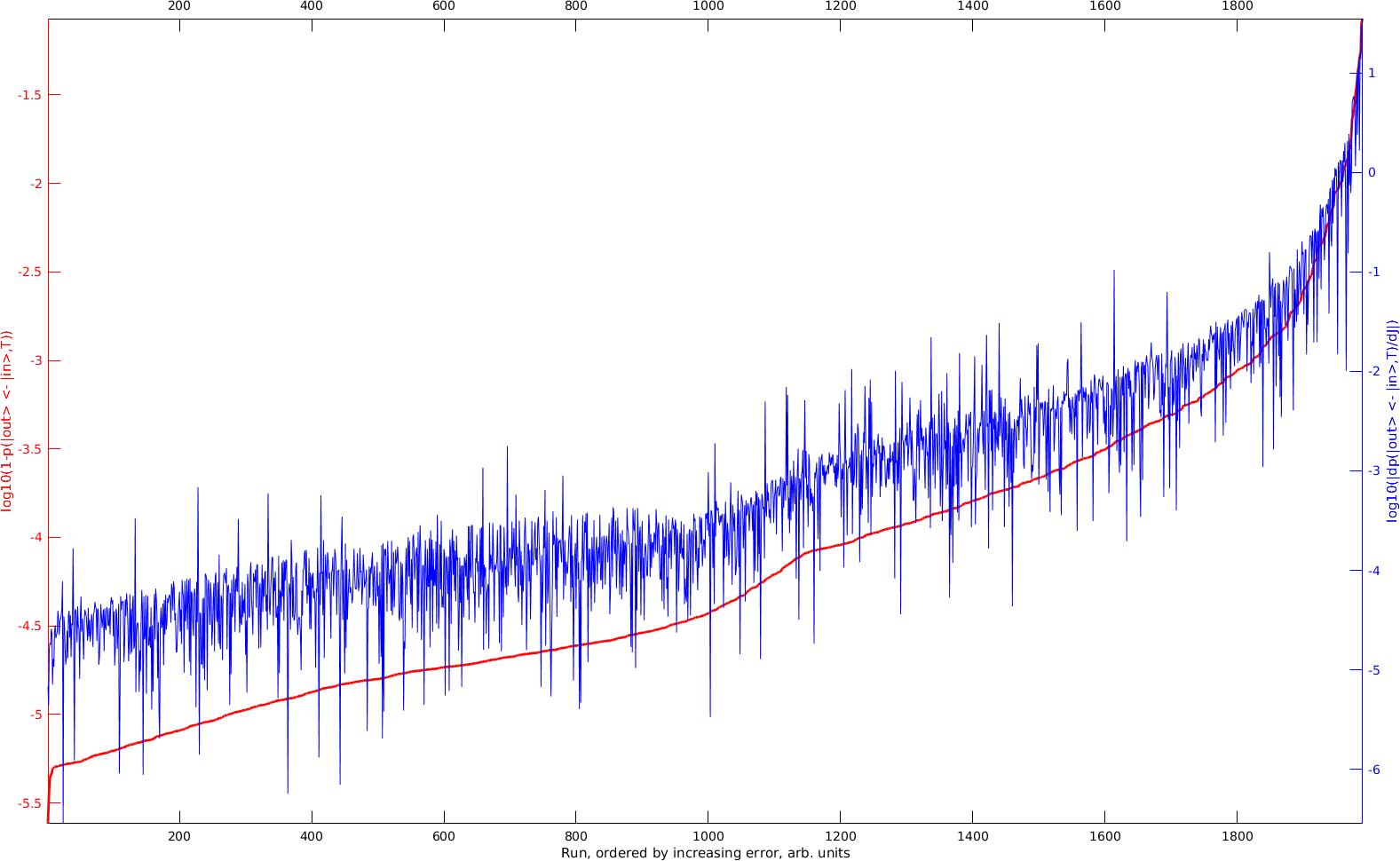}
  \caption{Logarithm of infidelity $1-\bar{p}$ (red) and logarithm of sensitivity (blue), ordered by increasing infidelity from left to right, of the localization controllers of a $14$-ring (left) and $19$-ring (right).}\label{fig:sensitivity_dt1}
\end{figure}

Since the system lacks asymptotic stability, the classical final value theorem does not hold. However,  generalized version of the Laplace final value theorem~\cite{formal_extended_final_value_Laplace,On_two_generalizations} allows us to compute the average error
\begin{align}
  &\lim_{T \to \infty} \frac{1}{T} \int_0^T E(t) \;dt = \lim_{s \to 0} s\mathcal{S}(s)\ket{n} \label{e:convolution}\\
  &= \left(I- \lim_{s \to 0} s\left(\mathcal{L}\left[e^{i \phi(t)}\right]\ast \left(sI+i \left(H_S+\widehat{D}\right)\right)^{-1}P\right)\right)\ket{n}. \notag
\end{align}

To get a better feeling for the Laplace convolution, assume, for simplicity, that perfect state transfer is achieved, i.e, $|\ip{n}{\Psi(t)}|=1$, in which case $e^{i\phi(t)}=\ip{n}{\Psi(t)}^*$. 
Then the convolution becomes
\begin{multline*}
  \mathcal{L}\left[e^{i\phi(t)}\right]\ast \left(sI+i\left(H_S+\widehat{D}\right)\right)^{-1}P\\
  = \bra{m} \left(sI-i\left(H_S+\widehat{D}\right)\right)^{-1} \ket{n} \ast \left(sI+i\left(H_S+\widehat{D}\right)\right)^{-1}P\\
  = \frac{1}{2\pi i} \int_{c-i\infty}^{c+i\infty} \bra{m} \left(zI-i\left(H_S+\widehat{D}\right)\right)^{-1} \ket{n} \\ \left((s-z)I+i\left(H_S+\widehat{D}\right)\right)^{-1}P\;dz,
\end{multline*}
where the path of integration is in the domain of convergence of $(sI+i(H_S+\widehat{D}))^{-1}$, i.e., $\operatorname{Re}(c)>0$. Therefore, $\lim_{s \to 0}$ should be interpreted as the limit as $s$ goes to zero, from the Right Half Plane.  The convolution is manageable via residue calculations, from which it follows that it has a pole at $s=0$, so that 
$\lim_{s \to 0} s \left(\mathcal{L}\left[e^{i\phi(t)}\right]\ast \left(sI+i\left(H_S+\widehat{D} \right) \right)^{-1}\ket{m}\right)$ 
is not trivially zero. 

\section{\label{sec:conclusion}Conclusions and Future Work}

We have shown that information transfer and localization in spin networks can be controlled by shaping the energy landscape using static potentials. This can be interpreted in terms of feedback control. However, it differs from measurement-based quantum feedback control in that the feedback is model-based and fully coherent. An advantage of this type of control is the relative simplicity, as neither measurements and state estimation nor rapidly modulated dynamic controls are required.  Furthermore, optimal feedback controllers are also the most robust, with superoptimal controllers simultaneously achieving perfect state transfer and vanishing sensitivity with respect to unavoidable uncertainties in the system.

In addition to exact-time transfer, additional gains in robustness can be obtained by optimizing the transfer to maximize the fidelity over a time window. This is especially important for practical applications as instantaneous readout requires effectively infinite bandwidth, which is usually unavailable. If the input and output states are identical, extending the time window yields solutions that achieve Anderson localization~\cite{Anderson-58,short_to_Anderson}, the closest equivalent to asymptotic closed-loop stability for Hamiltonian quantum networks.  Sensitivity properties of time-windowed optimized controllers  are analyzed from the statistical point of  view of  concordance between error and sensitivity as shown in Fig.~\ref{fig:sensitivity_dt} in~\cite{statistical_control}. Robustness under larger, combined initial preparation error and coupling error is available in~\cite{ssv_mu}.

Compared to dynamic control, finding optimal feedback control laws is considerably harder due to the complex optimization landscape. Analysis of the optimality conditions shows that the eigenstructure of the dynamic generators must satisfy certain symmetry conditions. Enforcing these conditions and careful choice of the initial values significantly improve the success rate of local optimization algorithms. Thus, enforcing constraints in this case improves performance of local optimization algorithms by simplifying the optimization landscape. For spin rings in particular, the constraints make them more similar to chains and the timing of the transmission peaks in the corresponding chains give good indications for the shortest possible transfer times in ring-based quantum routers.

From a control point of view, maximizing the transfer fidelity $|\bra{\Out}U_\vec{D}(T)\ket{\In}|=1$ is equivalent to canceling  the tracking error $|\ket{\Out}-e^{i \phi} U_\vec{D}(T)\ket{\In}|$, but the global phase factor means that we must think of the tracking error as an element of the complex projective space $\mathbb{C}\mathbb{P}^{N-1}$. Another difference to its classical counterpart is that our quantum feedback control scheme is not only $\ket{\Out}$-selective, but $\ket{\In}$-selective as well, while classical controllers are $\ket{\Out}$-selective as the target is specified by the reference signal, but the initial state is an equilibrium state.

There are many open questions for this control paradigm, ranging from the optimization landscape to global optimization algorithms that utilize the specific structure of the problem to find the best control laws. Furthermore, unlike dynamic control, for which explicit conditions for controllability in terms of the Lie algebra of the control operators are known, there are many theoretical questions in terms of attainability of the bounds and speed limits for selective information transfer.

\addtolength{\textheight}{-3.3cm}

\section*{Acknowledgments}
We would like to thank Robert Kosut and Aled Isaac for insightful discussions.

\appendices

\section{\label{app:signature-property}Proof of Full Signature Property for Rings with Symmetric Biases}

The argument of the proof is based on the symmetry of the biases in Eq.~\eqref{e:bias-symmetry}. To keep the notation simple, we write the eigenequation as $(H-\lambda I)x=0$.  Due to circular nearest-neighbor coupling for rings, it reads
\begin{equation*}
  x_{\ell-1 \bmod N}+(D_\ell-\lambda)x_{\ell \bmod N}+x_{\ell+1 \bmod N}=0, \quad \forall \ell.
\end{equation*}
We shall sometimes drop $\bmod N$ to simplify the notation. The key point is to rewrite the components of the eigenequation in symmetric pairs:
\begin{equation}
\label{e:pair}
\begin{split}
x_{m+\ell-1}+(D_{m+\ell}-\lambda)x_{m+\ell}+x_{m+\ell+1}&=0,\\
x_{n-\ell-1}+(D_{n-\ell}-\lambda)x_{n-\ell}+x_{n-\ell+1}&=0.
\end{split}
\end{equation}
Adding the equations and using the symmetry of the biases yields
\begin{equation}
\label{e:combined}
\begin{split}
(x_{m+\ell-1}+x_{n-\ell+1})+(D_{m+\ell}-\lambda)&(x_{m+\ell}+x_{n-\ell})\\
+&(x_{m+\ell+1}+x_{n-\ell-1})=0 .
\end{split}
\end{equation}
We must show that the sums of pairs of symmetrically related components vanish. This is achieved by writing the Eqs.~\eqref{e:pair}, \eqref{e:combined} for all $\ell$'s together with the ``boundary conditions'' allowing the equation to be solved by backsubstitution. By ``boundary conditions,'' we mean Eq.~\eqref{e:combined} involving $x_m+x_n=0$ together with Eq.~\eqref{e:pair} for some $\ell$ such that $\{m+\ell-1,m+\ell,m+\ell+1\} \cap \{n-\ell-1,n-\ell,n-\ell+1\} \ne \emptyset$.

How the two sets intersect and how to set up the corresponding boundary conditions depend on whether $|m-n|$ is even or odd. By symmetry we assume $m<n$.

\noindent\textbf{Case 1: If $n-m$ is odd,} then the recursion on the pairs of Eqs.~\eqref{e:pair} terminates at $\ell $ with $(n-\ell)=(m+\ell)+1 \bmod N$, and
\[ \overline{\ell}=\frac{n-m-1 \bmod N}{2}. \]
Observe that if $N$ is even the $\bmod N$ freedom yields two such $\ell$'s, defining two ring edges $(m+\overline{\ell})(n-\overline{\ell})$ in antipodal opposition.  To simplify the notation, define $\overline{m}:=m+\overline{\ell}$, $\overline{n}:=n-\overline{\ell}$, and let $X_{\overline{m}-\ell}$ to be the sum of the $\overline{m}-\ell$ component and its twin $\sigma(\overline{m}-\ell)=\overline{n}+\ell$, viz, $X_{\overline{m}-\ell}=x_{\overline{m}-\ell}+x_{\overline{n}+\ell}$. Writing the pair of Eqs.~\eqref{e:pair} for $\ell=\overline{\ell}$ and adding them in the combined Eq.~\eqref{e:combined} yields
\[ X_{\overline{m}-1}+(D_{\overline{m}}-\lambda +1)X_{\overline{m}}=0. \]
Defining the polynomial $p_{\overline{m}-1}(\lambda):=-(D_{\overline{m}}-\lambda +1)$ yields
\[ X_{\overline{m}-1}=p_{\overline{m}-1}(\lambda)X_{\overline{m}}. \]
Next, writing Eqs.~\eqref{e:pair}, \eqref{e:combined} for $\ell=\overline{\ell}-1$ yields
\[ X_{\overline{m}-2}+(D_{\overline{m}-1}-\lambda)X_{\overline{m}-1}+X_{\overline{m}}=0. \]
Writing $X_{\overline{m}-1}$ in its polynomial formulation yields $X_{\overline{m}-2}=p_{\overline{m}-2}(\lambda)X_{\overline{m}}$, where $p_{\overline{m}-2}(\lambda)=((D_{\overline{m}-1}-\lambda)(D_{\overline{m}}-\lambda+1)-1)$. The general equation should now be obvious:
\[ X_{\overline{m}-\ell}=p_{\overline{m}-\ell}X_{\overline{m}} \]
and the recursion on the polynomials is
\[ p_{\overline{m}-\ell}=-(D_{\overline{m}-\ell+1}-\lambda)p_{\overline{m}-\ell+1}-1. \]
Finally, we reach the situation where $X_m=p_m(\lambda)X_{\overline{m}}$. Since $X_m=x_m+x_n=0$, and if $p_m(\lambda)\ne 0$, we get $X_{\overline{m}}=0$, from where by backsubstitution $X_{\overline{m}-\ell}=0$ and the full symmetry is proved.

\noindent\textbf{Case 2: If $n-m$ is even,} then as $\ell$ increases, Eqs.~\eqref{e:pair}, \eqref{e:combined} terminate at
\[ \widehat{\ell}=\frac{n-m \bmod N}{2}, \]
with $\widehat{m}:= m+\widehat{\ell}=n-\widehat{\ell}=:\widehat{n}$. Observe that if $N$ is even the $\bmod N$ freedom yields two such $\widehat{m}=\widehat{n}$ at anti-podal points in the ring. The beginning of the recursion is a bit different  from the one of the odd case. We start with
\[ x_{\widehat{m}-1}+(D_{\widehat{m}}-\lambda)x_{\widehat{m}}+x_{\widehat{n}+1} =0 \]
and rewrite it as $X_{\widehat{m}-1}+(D_{\widehat{m}}-\lambda)x_{\widehat{m}}=0$. Defining the polynomial $p_{\widehat{m}-1}=-(D_{\widehat{m}}-\lambda)$ yields
\[ X_{\widehat{m}-1}=p_{\widehat{m}-1}(\lambda)x_{\widehat{m}}. \]
From here on the recursion is very much like the one of the odd case:
\[ X_{\widehat{m}-\ell}=p_{\widehat{m}-\ell}(\lambda)x_{\widehat{n}}, \]
together with the polynomial recursion
\[ p_{\widehat{m}-\ell}=(D_{\widehat{m}-\ell+1}-\lambda)p_{\widehat{m}-\ell+1}-2. \]
As $\ell$ increases, the recursion terminates as $X_m=p_m(\lambda)x_{\widehat{n}}$. Since $X_m=x_n+x_m=0$, and if $p_m(\lambda) \ne 0$, we get $x_{\widehat{n}}=0$ from where by backsubstitution $X_{\widehat{m}-\ell}=0$ and the theorem is proved.

\end{document}